\documentclass[12pt,reqno]{amsart}
\usepackage{amsmath,amsfonts,amssymb}
\usepackage[utf8]{inputenc}
\usepackage{mathrsfs}
\newcommand{\esssup}{\mathop{\mathrm{ess}\:\mathrm{sup}}\limits} 

\def\R{\mathbb R}

\def\Z{\mathbb Z}

\def\A{\mathbb A}
\def\B{\mathbb B}

\def\1{\mathbf 1}

\def\1{\bold 1}

\def\eps{\varepsilon}

\def\le{\leqslant}
\def\geq{\geqslant}
\def\ge{\geqslant}

\usepackage[height = 23.5cm, a4paper, hmargin={2.5cm , 2.0cm}]{geometry}

\usepackage{color}

\definecolor{darkred}{rgb}{0.9,0.1,0.1}

\theoremstyle{theorem}
\newtheorem{theorem}{Theorem}[section]
\newtheorem{proposition}[theorem]{Proposition}
\newtheorem{lemma}[theorem]{Lemma}

\newtheorem{remark}[theorem]{Remark}
\newtheorem{corollary}[theorem]{Corollary}

\numberwithin{equation}{section}

\begin{document}

\title[On operator estimates in   homogenization of nonlocal
operators]{On operator estimates in   homogenization of non-local
operators of convolution type}

\author[A.~Piatnitski,\  V.~Sloushch,\  T.~Suslina and E.~Zhizhina]{A.~Piatnitski${}^{1,2}$,\  V.~Sloushch${}^3$,\  T.~Suslina${}^3$ and E.~Zhizhina${}^2$}




\maketitle

\noindent
\begin{center}
{
\tiny{${}^1$ The Arctic university of Norway, campus Narvik,
P.O.Box 385, 8505 Narvik, Norway}\\[1.1mm]
${}^2$ Institute of Information Transmission Problems of RAS, Bolshoi Karetny, 19,
127051 Moscow, Russia\\[1.1mm]
${}^3$ St.~Petersburg State University, Universitetskaya nab. 7/9, 199034 St.~Petersburg, Russia\\[3mm]
}
\end{center}
\hskip 1.32cm
{\tiny{\it Email addresses}. \ \ A. Piatnitski:  {\tt apiatnitski@gmail.com}, \ \ V.~Sloushch: {\tt v.slouzh@spbu.ru}} \\[0.7mm]
${ }$\hskip 3.75cm
{\tiny{ T. Suslina: {\tt  t.suslina@spbu.ru}, \ \ E. Zhizhina: {\tt elena.jijina@gmail.com}  }}

\def\supind#1{${}^\mathrm{#1}$}


\begin{abstract}
The paper studies a bounded symmetric operator ${\mathbb A}_\eps$ in $L_2(\R^d)$ with
$$
({\mathbb A}_\eps u) (x) = \eps^{-d-2} \int_{\R^d} a((x-y)/\eps) \mu(x/\eps, y/\eps) \left( u(x) - u(y) \right)\,dy;
$$
here $\eps$ is a small positive parameter.
It is assumed that  $a(x)$ is a non-negative  $L_1(\R^d)$ function such that $a(-x)=a(x)$ and the moments
$M_k = \int_{\R^d} |x|^k a(x)\,dx$, $k=1,2,3$, are finite.  It is also assumed that $\mu(x,y)$
is $\Z^d$-periodic both in $x$ and $y$ function such that  $\mu(x,y) = \mu(y,x)$ and $0< \mu_- \leqslant \mu(x,y) \leqslant \mu_+< \infty$.
Our goal is to study the limit behaviour of the resolvent
 $({\mathbb A}_\eps + I)^{-1}$, as $\eps\to0$.
We show that, as $\eps \to 0$, the operator $({\mathbb A}_\eps + I)^{-1}$ converges in the operator
norm in $L_2(\R^d)$ to the resolvent $({\mathbb A}^0 + I)^{-1}$ of the effective operator ${\mathbb A}^0$
being a second order elliptic differential operator with constant coefficients of the form
${\mathbb A}^0= - \operatorname{div} g^0 \nabla$.
We then  obtain  sharp in order estimates of the rate of convergence.

\end{abstract}

\bigskip

\noindent


\section*{Introduction}

\subsection{Motivation: Contact model}
\label{subsec_01}
Evolution processes in the models of mathematical biology and population dynamics are often described in terms of
parabolic equation of the form $\partial_t u = -Au$ with a nonlocal convolution type operator $A$, the non-locality of $A$ reflects the fact that
the interaction in these models is nonlocal.  The kernel of $A$ has a multiplier $a(x-y)$ that specifies the intensity of interaction depending
on the distance and determines the localization properties of $A$.

One of the models of this type that has been considered in the existing literature is the so-called {\sl contact model} in $\mathbb R^d$, see \cite{KKP, KPMZh, KPZh}. This model relies on a continuous time Markov process that belongs to the class of death and birth processes
and is defined on the space $\Gamma$ of infinite locally finite configurations $\gamma\subset\mathbb R^d$.
 The behaviour of the process is determined by the intensity of birth and death.  On the one hand, each point $x\in\gamma$ can produce
 an offspring $y$ with the intensity $a(x-y)$, independently of other points, and we assume that  $\int_{\mathbb{R}^d} a(z) \,dz =1$.
 On the other hand, each point of $\gamma$ has a random life time, and the intensity of death is $m(x)>0$.  In the general case
 the intensities of birth and death might depend on the position in the space.
 The infinitesimal generator of such dynamics reads
 $$
LF(\gamma)  =  \sum_{x \in \gamma} \int_{\mathbb{R}^d} a(x-y)  \left( F(\gamma \cup y) - F(\gamma) \right) dy \ + \  \sum_{x \in \gamma} m(x) \left( F(\gamma\backslash x) - F(\gamma) \right).
$$
 The case of constant death intensity,  $m(x) \equiv \kappa$, has been studied in details in  \cite{KKP}, the so-called critical regime
  $m(x) \equiv 1$  being of special interest here. In this regime the process has a family of invariant measures.

The contact model has a remarkable property: the equation for the first correlation function describing the  density of
configuration is closed and not coupled with the equations for higher order correlation functions.
Notice that for the second and higher order correlation functions the corresponding evolutions  have a complicated hierarchical structure
that involves lower order correlation functions. The evolution equation for the first correlation function reads
\begin{equation}\label{PP}
\frac{\partial u}{\partial t} \, = \, -A u, \quad u = u (t, x), \quad x \in \mathbb{R}^d, \; t \ge 0, \quad u(0,x) = u_0 (x) \ge 0,
\end{equation}
where
\begin{equation*}
A u(x) \, = \,   m(x) u(x) \, - \, \int_{\mathbb{R}^d} a(x-y) u(y)\, dy.
\end{equation*}
If $m(x) \equiv 1$, the operator $A$ takes the form
\begin{equation}\label{LA}
A u(x)  =   u(x) \, - \, \int_{\mathbb{R}^d} a(x-y) u(y)\, dy \, = \, \int_{\mathbb{R}^d} a(x-y) (u(x) - u(y)) \,dy.
\end{equation}

The operator $A$ with a kernel $a(x-y)$ that depends only on the difference $x-y$ is a proper tool for describing space homogeneous media.
In order to model the processes in space inhomogeneous media it is natural to consider operators with more general kernels
of the form $a(x-y)\mu(x,y)$. These kernels depend both on the difference $x-y$ and on the starting and end points of jumps, the function $\mu(x,y)$ representing the local characteristics of the environment.
In particular, for the contact model in a periodic medium the operator $A$ in evolution equation   \eqref{PP} for the first correlation function
should be replaced with an operator
\begin{equation}\label{LA-per}
\mathbb{A} u(x)  =  \int_{\mathbb{R}^d} a(x-y) \mu(x,y) (u(x) - u(y)) dy
\end{equation}
with a periodic function $ \mu(x,y)$.

The study of large time behaviour of solutions to nonlocal evolution equations can be reduced to proper homogenization problems
for the corresponding nonlocal operators.  To illustrate this we consider a parabolic equation $\partial_t u = - Au$, where $A$  can be both second order elliptic differential operator and nonlocal convolution type operator. The natural way to avoid growing time intervals
is to multiply the temporal variable by a small positive parameter that we call $\eps^2$.  Then, to preserve the structure of the equation, one has to multiply the spatial variable by $\eps$, this change of variables being called the diffusive scaling.  In a non-homogeneous medium
this leads to  a family of equations that depends on a small parameter $\eps$, and we naturally arrive at the homogenization problem.

While the homogenization problems for differential operators have been actively studied for quite a long period, see for instance the monographs \cite{BaPa},
\cite{BeLP} and  \cite{JKO} and the bibliography therein,  similar problems for nonlocal convolution type operators with intregrable kernels
have not been considered in the existing literature till recently.  For the first time the periodic homogenization problem for such operators
was investigated in  \cite{PZh}, where it was shown that, under natural moment and coerciveness conditions, the limit operator is a second
order elliptic differential operator with constant coefficients and a positive definite effective matrix.  The homogenization procedure
relied on the corrector techniques, however the justification of convergence required new methods adapted to studying the nonlocal
operators. It was also proved in \cite{PZh} that the corresponding jump Markov process satisfies the invariance principle in the Skorokhod topology.  Similar homogenization problems in random media have been considered in  \cite{PiZh_jmpa}.
In the case of non-symmetric convolution type operators homogenization problem has been addressed in \cite{PiaZhi19}, it was shown
that  for the corresponding parabolic equations homogenization takes place in  moving coordinates.
Convolution type operators in periodically and randomly perforated domains  have been studied in \cite{BraPia22} and \cite{BraPia21}
by the variational methods.

A number of interesting  homogenization results for the Levy type nonlocal operators with non-integrable kernels have been obtained
in the recent works \cite{Sandr16}, \cite{Schwab10}, \cite{CCKW22}. However, the properties of these operators differ a lot from the properties of operators studied in this
work.

\subsection{Estimates for the rate of convergence in homogenization theory}
\label{subsec_02}
Essential interest in obtaining various estimates for the rate of convergence in the mathematical homogenization theory
 arose shortly after emergence of this theory itself. This interest was stimulated by  important applications of  this theory in such disciplines as
 composite materials,  photonic crystals, porous media, network constructions and many others.   In addition, obtaining  these estimates turned out to be an interesting mathematical challenge.

For periodic media the first results in this direction have been obtained in  \cite{Kesa, Kesa2} and \cite{Vanni}.
The works   \cite{Kesa, Kesa2} dealt with boundary value problems for elliptic operators with oscillating coefficients,
for these problems $L_2$ estimates for the rate of convergence were proved.
Then  in \cite{Vanni} similar estimates have been obtained for elliptic operators in
periodically perforated domains including spectral problems with the Steklov condition
on the perforation border.

Later on many authors contributed to this topic, a large number of interesting results on the rate of convergence
for solutions to elliptic and evolution differential equations and systems of equations have been obtained both in periodic and locally periodic environments. Among others, the elasticity system and, in some particular cases, the Maxwell system have been   investigated.
We refer here only to some books \cite{BaPa, BeLP, JKO}, where these questions have been addressed.

For homogenization problems in random statistically homogeneous environments the first estimates for the rate of convergence
 were obtained in   \cite{Yur}. Further progress in this field was achieved in the recent works  \cite{GloOt2, GloOt1}.

Notice however that in all the above mentioned papers the estimates for the rate of convergence were obtained in the strong topology,
not the operator topology, the constants in these estimates depended on the regularity of the right-hand sides.

In the pioneer works  of M.~Birman and T.~Suslina  \cite{BSu1, BSu3, BSu4} a new approach to homogenization problems in periodic media has been developed, it relies on a version of the spectral method. With the help of this method  for a wide class of homogenization problems
estimates in the operator topology were justified.

In order to clarify this method, let us consider a scalar elliptic operator of the form
$A_\eps = - \operatorname{div} g(x/\eps) \nabla$, $\eps>0$, in $L_2(\R^d)$
with a positive definite bounded $\Z^d$-periodic matrix $g(\cdot)$.
Classical result states that for any $F\in L_2(\mathbb R^d)$, as $\eps\to0$, a solution $u_\eps$ of the equation    $$
(A_\eps u_\eps)(x) + u_\eps(x) = F(x),\quad x \in \R^d,
$$
converges in $L_2(\R^d)$ to a solution $u_0$ of the homogenized equation
that reads
$$
(A^0 u_0)(x) + u_0(x) = F(x),\quad x \in \R^d;
$$
here $A^0 = - \operatorname{div} g_{\operatorname{hom}} \nabla$ is the so-called effective operator and  $g_{\operatorname{hom}}$ is a constant positive definite matrix.
Moreover, the estimate
$$
\| u_\eps - u_0\|_{L_2(\R^d)} \leqslant C(F) \eps
$$
holds true.
In \cite{BSu1} a stronger estimate has been obtained. Namely, it was shown that
$$
\| u_\eps - u_0\|_{L_2(\R^d)} \leqslant C \eps \|F\|_{L_2(\R^d)}
$$
with a constant $C$ that does not depend on $F$.
This result can be reformulated in terms of operator convergence: as $\eps\to0$, the resolvent
$(A_\eps +I)^{-1}$ converges in the operator norm in  $L_2(\R^d)$ to $(A^0 +I)^{-1}$,
and the following order-sharp estimate is fulfilled:
\begin{equation}
\label{BSu1}
\| (A_\eps +I)^{-1} - (A^0 +I)^{-1}\|_{L_2(\R^d) \to L_2(\R^d)} \leqslant C \eps.
\end{equation}
In the homogenization theory the inequalities of this type are called operator estimates of approximation error.
Then in \cite{BSu3} a more advanced uniform operator approximation of the resolvent $(A_\eps +I)^{-1}$ in
$L_2(\mathbb R^d)$ has been found, the estimate derived in this work takes into account the corrector
and provides the precision of order $O(\eps^2)$. As was shown in \cite{BSu4} similar approximation  yields
the precision of order $O(\eps)$ in the space of operators acting from $L_2(\mathbb R^d)$ to $H^1(\R^d)$.

Analogous results for parabolic semigroups  $e^{-A_\eps t}$, $t>0$, have been obtained in  \cite{Su1, Su2, V, Su3}.
It should be emphasized that the above mentioned results on the estimates in operator topology are valid for a wide class of second order  matrix elliptic operators.

The spectral method used in the cited collection of works relies on scaling transformation, Floquet-Bloch theory and analytic perturbation theory.

We illustrate this method by explaining the main ideas of obtaining  the estimate in \eqref{BSu1}.
To explain the main ideas, we derive the estimate in \eqref{BSu1}. 
\begin{equation}
\label{BSu2}
\| (A + \eps^2 I)^{-1} - (A^0 + \eps^2 I)^{-1}\|_{L_2(\R^d) \to L_2(\R^d)} \leqslant C \eps^{-1},
\end{equation}
 where $A = - \operatorname{div} g(x) \nabla = D^* g(x) D$, $D = -i \nabla$. Applying then the unitary Gelfand transform we represent
 the operator $A$ as a direct integral with respect to the operators  $A(\xi)$ acting in  $L_2(\Omega)$
and depending on a parameter $\xi \in \widetilde{\Omega}$  (quasi-momentum); here $\Omega = [0,1)^d$ is the periodicity cell $\mathbb R^d/\mathbb Z^d$, and  $\widetilde{\Omega} = [-\pi,\pi)^d$ is the cell of the dual lattice. The operator $A(\xi)$  is of the form $A(\xi) = (D + \xi)^* g(x) (D+ \xi)$ and equipped with the periodic boundary conditions. Estimate \eqref{BSu2} is a consequence of the following one:
\begin{equation*}
\| (A(\xi) + \eps^2 I)^{-1} - (A^0(\xi) + \eps^2 I)^{-1}\|_{L_2(\Omega) \to L_2(\Omega)} \leqslant C \eps^{-1},
\quad \xi \in \widetilde{\Omega}.
\end{equation*}
The family of operators $\{A(\xi)\}$ is analytic, and all the operators of this family have
a compact resolvent. Thus the methods of analytic perturbation theory apply, and
the resolvent  \hbox{$(A(\xi) + \eps^2 I )^{-1}$} can be approximated in terms of the spectral characteristics
of this operator family at the spectral edge.  In particular, the effective matrix is 
proportional with the coefficient $\frac12$ to the Hessian
of the first eigenvalue $\lambda_1(\xi)$ of the operator $A(\xi)$ at $\xi=0$.
Therefore, the homogenization can be treated as a spectral threshold effect at the edge
 of the spectrum of an elliptic operator.

Later on with the help of the spectral method various operator estimates for the rate of convergence were obtained for a wide class of elliptic
and evolution equations with periodic coefficients, among them are non-stationary Schr\"odinger-type and hyperbolic equations,
see \cite{BSu5, M, Su4, D, DSu1, DSu2}, stationary and non-stationary Maxwell system, see   \cite{Su_Maxwell, Su_Maxwell2, DSu3},
and many others.

Another approach to estimating the operator norm of the discrepancy in the homogenization procedure, the so-called ''shift method'',  was developed by V.~Zhikov and S.~Pastukhova in \cite{Zh, ZhPas1, ZhPas2}, see also the survey paper  \cite{ZhPas3} and references therein.
This approach is based on the analysis of the first approximation of the solution and on introducing an additional variable that takes on values in
$\Omega$. It allows to deal with homogenization of boundary-value problems in bounded domains in $\mathbb R^d$.

The works \cite{Zh, ZhPas1} studied homogenization problems both for elliptic operators defined
in the whole $\mathbb R^d$ and for boundary-value problems with the Dirichlet and Neumann boundary conditions.
It was shown that the discrepancy admits the estimate of order $O(\eps^{1/2})$ in the uniform operator topology in
$L_2$ and, in the presence of correctors, in the operator norm $L_2 \to H^1$. The estimates are not as good as in the case
of the whole space due to the influence of the domain boundary.
Similar results for the Dirichlet and Neumann problems for a scalar operator in a bounded domain have been obtained in \cite{Gr1, Gr2}  by means of the unfolding method.
In \cite{Gr2} for the first time an estimate of the sharp order $O(\eps)$ was justified in the operator norm in $L_2$.

For elliptic systems similar estimates have been proved in \cite{KeLiS}
and \cite{PSu, Su_Math, Su_SIAM},  the case of initial-boundary problems
for the corresponding parabolic equations has been treated in \cite{MSu, GeS},
and Maxwell system in a bounded domain was considered in \cite{Su_Maxw3, Su_Maxw4}.

In recent years further progress has been achieved in this topic. We quote here \cite{Gu}, where operator estimates were obtained for the Stokes system,  \cite{Ve, KukSu, Pas1, Pas2, SlSu, MilSu, Su_Ari, Su_Lob} that dealt with elliptic operators of higher order in $\mathbb R^d$,
and   \cite{Su_high1, Su_high2, Su_high3},  where the higher order elliptic operators in a bounded domain were considered.

The works \cite{Bor, PasT1, PasT2, Se1, Se2, Se3} focused on operator estimates for operators with locally periodic
coefficients, the high contrast problems were investigated in  \cite{ChCo, ChKis, ChErKis}.
Similar questions can be raised when boundary homogenization problems are considered. Some of them have been addressed
in \cite{BorCFR, BorSh, BorCD}.

However, it should be emphasized that in all these papers the operator estimates have been obtained for differential operators.
To our best knowledge the case of nonlocal operators has not been examined in the existing literature.

\subsection{Problem setup. Main result}
\label{subsec_03}
The present work studies convolution type operators that read
 \begin{equation}
 \label{Sus1}
 \mathbb{A}_\eps u(x)=\frac1{\eps^{d+2}}\int_{\mathbb R^d} a\Big(\frac{x-y}{\eps}\Big)\mu\Big(\frac x\eps,\frac y\eps\Big)\big( u(x)-u(y)\big)\,dy,\ \ x\in\mathbb{R}^{d},\ \ u\in L_{2}(\mathbb{R}^{d}),
 \end{equation}
here $\eps$ is a small positive parameter,  $a(x)$ is an even non-negative function such that  $a(\cdot)\in L_{1}(\mathbb{R}^{d})$ and $\|a\|_{L_{1}}>0$, $\mu(x,y)$ is a $\mathbb{Z}^{d}$ periodic in $x$ and $y$ function such that
$\mu(x,y)=\mu(y,x)$ and $0<C_1\leqslant \mu(x,y)\leqslant C_2$ for some constants $C_1$ and $C_2$.
Under these conditions the operator $\mathbb{A}_{\eps}$ is bounded, self-adjoint and non-negative in $L_2(\mathbb R^d)$.
We assume furthermore that the moments $M_{k}=\int_{\mathbb{R}^{d}}|x|^{k}a(x)dx$ are finite
for $k=1,2,3$.

As was mentioned above,  the homogenization problem for operators defined in \eqref{Sus1} was solved in  \cite{PZh}.
It was shown that the resolvent $(\mathbb{A}_{\varepsilon}+I)^{-1}$ converges, as $\eps\to0$, in the strong operator topology
to the resolvent $(\mathbb{A}^{0}+I)^{-1}$ of the effective operator being the second order elliptic operator
of the form   $\mathbb{A}^{0}=-\operatorname{div}g^{0}\nabla$. It is interesting to observe that although
the operators  $\mathbb{A}_{\eps}$ are bounded and nonlocal for any $\eps>0$, the limit operator is unbounded
and local.


Although the estimates for the rate of convergence in the just mentioned homogenization problem seem not to be presented in the mathematical literature,  for sufficiently regular right-hand sides such estimates can be obtained with the help of asymptotic expansions constructed in~\cite{PZh}.

Our goal here is to prove the convergence of the resolvent $(\mathbb{A}_{\varepsilon}+I)^{-1}$
to the resolvent $(\mathbb{A}^{0}+I)^{-1}$  in the operator norm  topology in $L_2(\mathbb R^d)$
and to derive order-sharp  estimates for the discrepancy.


In order to formulate the main result of this paper, we first recall the definition of the effective matrix  $g^0$.
Consider an auxiliary cell problem that reads: find a $\mathbb{Z}^{d}$-periodic vector-valued solution of the equation
\begin{equation}\label{Sloushch1}
\left\{
\begin{array}{l}
\int_{\mathbb{R}^{d}}a(x-y)\mu(x,y)(v(x)-v(y))\, dy=\int_{\mathbb{R}^{d}}a(x-y)\mu(x,y)(x-y)\, dy,\\[2mm]
\int_{\Omega}v(y)dy=0;
\end{array}
\right.
\end{equation}
here $\Omega:=[0,1)^{d}$ is a periodicity cell. This problem has a unique solution.
The effective matrix $g^{0}=\frac{1}{2}\{g_{ij}\}_{i,j=1,\dots,d}$ is defined by
\begin{equation*}
g_{ij}=\int_{\Omega}dx\int_{\mathbb{R}^{d}} dy\Big((x_{i}-y_{i})(x_{j}-y_{j})-v_{j}(x)(x_{i}-y_{i})-v_{i}(x)(x_{j}-y_{j})\Big)a(x-y)\mu(x,y),\
\ i,j=1,\dots,d.
\end{equation*}
It turns out that the matrix $g^{0}$ is positive definite, see \cite{PZh}.
The domain of the homogenized operator coincides with the Sobolev space $H^2(\mathbb R^d)$.
Our main result is Theorem \ref{teor3.1} which states that the following estimate holds true:
\begin{equation}\label{Sloushch2}
\|(\mathbb{A}_{\varepsilon}+I)^{-1}-(\mathbb{A}^{0}+I)^{-1}\|_{L_{2}(\mathbb{R}^{d})\to
L_{2}(\mathbb{R}^{d})}\leqslant C(a,\mu)\varepsilon,\ \ \varepsilon>0.
\end{equation}
This estimate is sharp in order, and the constant $C(a,\mu)$ can be calculated explicitly.

\subsection{The methods}
\label{subsec_04}

To approximate the resolvent of operator \eqref{Sus1} and justify the formulated above result we modify and adapt the spectral method
discussed in Section \ref{subsec_02}.

The first two steps, the scaling transformation and the unitary Gelfand transform allowing to represent
the operator in \eqref{LA-per} as a direct integral with respect to the family $\mathbb{A}(\xi)$, remain unchanged.
 However, in contrast to the case of differential operators, the family $\mathbb{A}(\xi)$ depending on the parameter
$\xi \in \widetilde{\Omega}$ need not be analytic in $\xi$, thus the methods of analytic perturbation theory do not
apply to this family. Instead, our approach relies on a $C^3$ regularity of this family which is granted by the existence of the first three moments
of function $a(z)$.

Let us explain some details of this approach. It is known that, in order to approximate the resolvent
$(\mathbb{A}(\xi)+\eps^{2}I)^{-1}$ for small  $\eps>0$, it suffices to construct the asymptotics of the operator-function   $\mathbb{A}(\xi)F(\xi)$, $\xi\to0$, where $F(\xi)$ is a spectral projector of  the operator  $\mathbb{A}(\xi)$ 
that corresponds to  some neighbourhood of zero.

In the previous works that dealt with the differential operators the asymptotics of the operator $\mathbb{A}(\xi)F(\xi)$ for small  $|\xi|$
relied on the asymptotics of the principal eigenvalue  $\lambda_1(\xi)$ of operator $A(\xi)$.
In the present paper, in addition to this classical method, we provide two alternative ways of computing the asymptotics of
$\mathbb{A}(\xi)F(\xi)$ as $\xi\to0$ which are not based on constructing the asymptotics of  $\lambda_1(\xi)$.
One of them that looks promising for applications in the homogenization theory relies on computing the coefficients of the asymptotics
of  $\mathbb{A}(\xi)F(\xi)$, $\xi\to0$, by means of integrating the resolvent $(\mathbb{A}(0)-\zeta I)^{-1}$ as well as some operator-functions of this resolvent, over appropriate contours, see Section 4.2.

As in the case of differential operators, the key role in the studied homogenization problem is played by the spectral characteristics of
operator $\mathbb{A}$ near the lower edge of its spectrum (the so-called threshold characteristics).  Thus, the homogenization of nonlocal
operator can also be considered as a threshold effect at the spectrum edge.

\subsection{Plan of the paper}
The paper consists of Introduction, three sections and Appendix.
In Section \ref{sec_1} we introduce  operator  ${\mathbb A}$, represent this operator as a direct integral over the family ${\mathbb A}(\xi)$
and obtain lower bounds for the quadratic form of the operator ${\mathbb A}(\xi)$.

In  Section 2 the threshold characteristics of the operator family
${\mathbb A}(\xi)$ are studied in the vicinity of the spectrum bottom. Then we approximate the resolvent
$(\!{\mathbb A}(\xi) + \eps^2 I)^{-1}$ for small~$\eps$.

Finally, in Section 3 we complete the proof of the main result of this work.

The Appendix contains a number of auxiliary statements.

\subsection{Notation}
\label{subsec_06}

The norm in a normed space $X$ is denoted by $\|\cdot\|_{X}$, the index $X$ is dropped if this
does not lead to an ambiguity. If $X$ and $Y$ are linear normed space, the norm of a linear
bounded operator $T:X\to Y$ is denoted $\|T\|_{X\to Y}$ or just $\|T\|$.
The notation $\mathcal{B}(X)$ is used for the space of bounded linear operators in a normed space $X$.
Given a collection of vectors $F\subset X$, we denote its span  by $\mathcal{L}\{F\}$.
An open ball in a normed space centered at $x_{0}$ of radius $r>0$ is denoted by  $B_{r}(x_{0})$.

The spectrum and the essential spectrum of a self-adjoint operator $\A$ in a Hilbert space $\mathfrak{H}$
are denoted by $\sigma(\A)$ and $\sigma_{\text{e}}(\A)$, respectively.

For a Lebesgue measurable set $\mathcal O\subset \R^d$ of positive Lebesgue measure
we denote by $L_{p}({\mathcal O})$, \hbox{$1 \le p \le \infty$}, the standard $L_p$ space
of functions defined on $\mathcal O$.
The inner product in  $L_{2}({\mathcal O})$ is denoted  by $(\cdot,\cdot)_{L_{2}({\mathcal O})}$,
while for the scalar product in $\R^d$ and $\mathbb{C}^{d}$ the notation  $\langle\cdot,\cdot\rangle$ is used.
Symbol $\mathcal{S}(\R^{d})$ stands for the Schwartz class of functions in $\R^{d}$.
The characteristic function of a set $E\subset\R^d$ is denoted by $\mathbf{1}_{E}$.

\section{Nonlocal Schr\"odinger operator: \\ representation as a direct integral and estimates}
\label{sec_1}

\subsection{Operator $\A(a,\mu)$}
\label{subsec_11}
Given real-valued functions $a\in L_{1}(\R^d)$ and $\mu\in L_{\infty}(\R^d\times \R^d)$ we define a nonlocal Schr\"odinger operator $\A = \A(a,\mu)$ in the space  $L_{2}(\R^d)$, $d\ge 1$, by
\begin{equation*}
\A(a,\mu) u(x):=\int_{\R^d}a(x-y)\mu(x,y)(u(x)-u(y)) \,dy,\ \ x\in\R^d.
\end{equation*}
Rearranging the expression on the right-hand side,  one has $\A(a,\mu)=p(x;a,\mu)-\B(a,\mu)$ with
\begin{equation*}
\begin{gathered}
p(x;a,\mu):=\int_{\R^d}a(x-y)\mu(x,y) \,dy,\ \ x\in\R^d,\\
\B(a,\mu) u(x):=\int_{\R^d}a(x-y)\mu(x,y)u(y)\,dy,\ \ x\in\R^d;
\end{gathered}
\end{equation*}
here and in what follows we identify the function $p(\cdot;a,\mu)$ and the operator of multiplication by this function.
According to the Schur test, see \cite[Theorem 5.2]{Halmos78}, the operator   $\B(a,\mu):L_{2}(\R^d)\to L_{2}(\R^d)$ is bounded in $L_2(\mathbb R^d)$
and its norm satisfies the upper bound $\|\B(a,\mu)\|\le\|\mu\|_{L_\infty}\|a\|_{L_1}$.
For the reader convenience we  prove this statement in the Appendix, see Lemma \ref{lemma4.1}.
The potential $p(x;a,\mu)$ is also bounded
and admits the estimate $\|p(\cdot;a,\mu)\|_{L_\infty}\le\|\mu\|_{L_\infty}\|a\|_{L_1}$. Therefore, the operator
$\A(a,\mu):L_{2}(\R^d)\to L_{2}(\R^d)$ is bounded.
It is convenient to introduce, in addition to $\A(a,\mu)$, an auxiliary operator $\A_{0}(a):=\A(a,\mu_{0})$ with $\mu_{0}\equiv 1$.
Letting $\B_{0}(a):=\B(a,\mu_{0})$ and $p_{0}(x;a):=p(x;a,\mu_{0})$ we observe that $\B_{0}(a)$ is a convolution operator with the kernel
$a(\cdot)$, and the potential $p_{0}(x;a)$ is constant.


From now on we assume that the functions $a$ and  $\mu$ possess the following properties:
\begin{gather}
\label{e1.1}
a\in L_{1}(\R^d),\ \ \operatorname{mes}\{x\in\R^d:a(x)\not=0\}>0,\ \ a(x)\ge 0,\ \ a(-x)=a(x),\ \ x\in\R^d;
\\
\label{e1.2}
0<\mu_{-}\le\mu(x,y)\le\mu_{+}<+\infty,\ \ \mu(x,y)=\mu(y,x),\ \ x,y\in\R^d;
\\
\label{e1.3}
\mu(x+m,y+n)=\mu(x,y),\ \ x,y\in\R^d,\ \ m,n\in\Z^d;
\\
\label{e1.4}
M_{k}(a):=\int_{\R^d}|x|^{k}a(x)dx<+\infty,\ \ k=1,2,3.
\end{gather}
Since the coefficients $a(\cdot)$ and $\mu(\cdot)$ are real-valued, the potential  $p(x)=p(x;a,\mu)$ is also real.   Moreover, under conditions  (\ref{e1.1}) and (\ref{e1.2}) the potential $p(x;a,\mu)$ satisfies the estimate
\begin{equation}
\label{e1.5}
\mu_{-}\|a\|_{L_1(\R^d)} \le p(x)\le\mu_{+}\|a\|_{L_1(\R^d)},\ \ x\in\R^d,
\end{equation}
and the operator $\B(a,\mu)$ is bounded and self-adjoint.  Thus
 $\A(a,\mu)$ is bounded and self-adjoint as well.

\subsection{Quadratic form of operator $\A(a,\mu)$}
\label{subsec12}
Under assumptions \eqref{e1.1} and \eqref{e1.2} the quadratic form of operator $\A(a,\mu)$ admits the following representation
(see, for instance, \cite{KPMZh}):
\begin{equation}\label{e1.6}
(\A(a,\mu) u,u)=\frac{1}{2} \int_{\R^d} \int_{\R^d} dx\,dy\, a(x-y)\mu(x,y) |u(x)-u(y)|^{2},\ \ u\in L_{2}(\R^d).
\end{equation}
 Indeed, from the relations $a(x-y) = a(y-x)$ and $\mu(x,y) = \mu(y,x)$ we obtain
\begin{multline*}
(\A(a,\mu) u,u)=\intop_{\R^d}dx\,\overline{u(x)}\intop_{\R^d}dy\,a(x-y)\mu(x,y)(u(x)-u(y))=
\\
=\intop_{\R^d}\intop_{\R^d}dx\,dy\, a(x-y)\mu(x,y)|u(x)-u(y)|^{2}+\intop_{\R^d}dy\,\overline{u(y)} \intop_{\R^d}dx\, a(x-y)\mu(x,y)(u(x)-u(y))=
\\
=\intop_{\R^d}\intop_{\R^d}dx\,dy\, a(x-y)\mu(x,y)|u(x)-u(y)|^{2}-(\A(a,\mu) u,u),\ \ u\in L_{2}(\R^d).
\end{multline*}
This yields  (\ref{e1.6}).
Considering representation \eqref{e1.6} we conclude that the operator  $\A(a,\mu)$  is non-negative
and the estimates
\begin{equation}
\label{e1.7}
\mu_{-}(\A_{0}(a) u,u)\le (\A(a,\mu) u,u)\le\mu_{+}(\A_{0}(a) u,u),\ \ u\in L_{2}(\R^d),
\end{equation}
hold true.  Since $\B_{0}(a)$ is a convolution operator, in the space of Fourier images of functions from $L_2(\mathbb R^d)$ it
acts as the operator of multiplication by the function $\widehat a(\xi)$, $\xi\in\R^d$, with
\begin{equation*}
\widehat a(\xi) := \int_{\R^d}e^{-i \langle \xi, z \rangle}a(z)\, dz,\ \ \xi\in\R^d.
\end{equation*}
Notice that due to \eqref{e1.1} the function $\widehat a(\xi)$ is continuous and tends to zero at infinity.
Consequently, the operator $\A_{0}(a)$ is unitary equivalent to the operator of multiplication by the function $\widehat a(0)-\widehat a(\xi)$,
and $\lambda_{0}=0$ belongs to the spectrum of  $\A_{0}(a)$.
Since $\A_{0}(a)$ is a non-negative operator, $\lambda_{0}$ is its spectral edge.
In view of estimates \eqref{e1.7} the point $\lambda_{0}=0$ is also the lower edge of the spectrum of
operator  $\A(a,\mu)$.

\subsection{Representation of $\A(a,\mu)$ as a direct integral}
\label{subsec_13}
Due to conditions \eqref{e1.1}--\eqref{e1.3} the operator of multiplication by $p(x;a,\mu)$ and the operator $\B(a,\mu)$ commute
with the operators $S_n$ defined by
\begin{equation*}
S_{n}u(x) =u(x+n),\ \ x\in\R^d,\ \ n\in\Z^d.
\end{equation*}
So does $\A(a,\mu)$.
Thus $\A(a,\mu)$ and $\B(a,\mu)$ are periodic operators with a periodicity lattice  $\Z^d$. Denote by $\Omega:=[0,1)^{d}$ the corresponding
periodicity cell and by $\widetilde\Omega:=[-\pi,\pi)^{d}$ the periodicity cell of the dual lattice  $(2\pi\mathbb{Z})^{d}$.

Let us recall the definition of the Gelfand transform, see, for example,  \cite{Sk} or \cite[Chapter~2]{BSu1}. We call it $\mathcal{G}$.
For functions $u$ from the Schwartz class  ${\mathcal S}(\R^d)$ it is defined as follows:
\begin{equation*}
\mathcal{G}u(\xi,x):=(2\pi)^{-d/2}\sum_{n\in\Z^d}u(x+n)e^{-i \langle \xi, x+n\rangle},\
\ \xi\in\widetilde\Omega,\ \ x\in\Omega,\ \ u\in {\mathcal S}(\R^d).
\end{equation*}
Since $\|\mathcal{G}u\|_{L_{2}(\widetilde\Omega\times\Omega)}=\|u\|_{L_{2}(\R^d)}$, then $\mathcal{G}$ can be extended by continuity
to a unitary mapping from $L_{2}(\R^d)$ to $\int_{\widetilde\Omega}\bigoplus L_{2}(\Omega)d\xi=L_{2}(\widetilde\Omega\times\Omega)$.

In what follows we assume that for any $v\in\int_{\widetilde\Omega}\bigoplus L_{2}(\Omega)d\xi$ and any
 $\xi\in\widetilde\Omega$ the function $v(\xi,\cdot)$ is extended periodically from the cell $\Omega$ to the whole
$\R^d$.
Since the operators  $\A(a,\mu)$ and $\B(a,\mu)$ are periodic, the Gelfand transform partially diagonalizes them.
More precisely, the operators $\mathcal{G}\A(a,\mu) \mathcal{G}^{*}$ and $\mathcal{G}\B(a,\mu) \mathcal{G}^{*}$ take the form
\begin{equation}
\label{e1.8}
\mathcal{G}\A(a,\mu) \mathcal{G}^{*}u(\xi,x)=\A(\xi,a,\mu)u(\xi,x),\ \
\mathcal{G}\B(a,\mu) \mathcal{G}^{*}u(\xi,x)=\B(\xi,a,\mu)u(\xi,x),\ \
\xi\in\widetilde\Omega,\ \ x\in\Omega;
\end{equation}
here the operators $\A(\xi,a,\mu)$ and $\B(\xi,a,\mu)$ are bounded self-adjoint operators in $L_{2}(\Omega)$
that are defined by
\begin{equation*}
\A(\xi,a,\mu)=p(x;a,\mu)-\B(\xi,a,\mu),\ \
\B(\xi,a,\mu)u(x)=\int_{\Omega}\widetilde a(\xi,x-y)\mu(x,y)u(y)\,dy,\ \ u\in L_{2}(\Omega);
\end{equation*}
\begin{equation*}
\widetilde{a}(\xi,z) :=\sum_{n\in\Z^d}a(z+n)e^{-i \langle \xi, z+n \rangle },\ \
\xi\in\widetilde\Omega,\ \ z\in\R^d.
\end{equation*}
We emphasize that $p(x;a,\mu)$ does not depend on $\xi$ and satisfies the relation
\begin{equation}\label{e1.9}
p(x;a,\mu)=\int_{\Omega}\widetilde a(0,x-y)\mu(x,y)\,dy.
\end{equation}
In view of $\Z^d$-periodicity of function $\widetilde a(\xi,\cdot)$ the estimate
$\int_{\Omega}|\widetilde a(\xi,z)|\,dz\le\|a\|_{L_1}$
implies that
$\widetilde a(\xi,\cdot)\in L_{1,\mathrm{loc}}(\R^d)$  and
$$
\|\widetilde a(\xi,\cdot)\|_{L_{1}([-1,1]^{d})}\le 2^{d}\|a\|_{L_1(\R^d)}, \quad \xi\in\widetilde\Omega.
$$
Therefore, by Corollary \ref{corollary4.2}, the operator $\B(\xi,a,\mu)$ is compact, and the following estimate
holds
\begin{equation*}
\|\B(\xi,a,\mu)\|\le\mu_{+}\intop_{[-1,1]^{d}}|\widetilde
a(\xi,z)|dz\le 2^{d}\mu_{+}\|a\|_{L_1(\R^d)},\ \ \xi\in\widetilde\Omega.
\end{equation*}
We conclude that the essential spectrum of the operator  $\A(\xi,a,\mu)$ coincides with
$\mathrm{ess}\text{-}\mathrm{Ran}\,p(\cdot;a,\mu)$ and, due to the lower bound in \eqref{e1.5},
for any $\xi\in\widetilde\Omega$ the spectrum of   $\A(\xi,a,\mu)$  in the interval $(-\infty,\mu_{-}\|a\|_{L_1} )$
is discrete.
The goal of this section is to estimate the lower edge of the spectrum of operator $\A(\xi,a,\mu)$ for
all $\xi$.

The next statement provides a representation of the quadratic form of the operator $\A(\xi,a,\mu)$ that will
be convenient for the further analysis.

\begin{lemma}
\label{lemma1.1}
Under conditions  \eqref{e1.1}--\eqref{e1.3} for any $\xi\in\widetilde\Omega$ the following relation holds:
\begin{equation}
\label{e1.11}
(\A(\xi,a,\mu)u,u) = \frac{1}{2}\intop_{\Omega}dx\intop_{\R^d}\,dy\,
a(x-y)\mu(x,y)|e^{i\langle \xi, x\rangle}u(x)-e^{i \langle \xi, y\rangle}u(y)|^{2},\ \ u\in L_{2}(\Omega);
\end{equation}
it is assumed here that the function $u\in L_{2}(\Omega)$  is extended periodically to the whole $\R^d$.
\end{lemma}

\begin{proof}
By \eqref{e1.9} for any $x\in\Omega$ we have
\begin{equation}\label{e1.12}
\A(\xi,a,\mu)
u(x)=\int_{\Omega} \left( \widetilde{a}(0,x-y)u(x)-\widetilde{a}(\xi,x-y)u(y)\right) \mu(x,y)\, dy,\
\ \ u\in L_{2}(\Omega).
\end{equation}
Identifying a function  $u\in L_{2}(\Omega)$ with its $\Z^d$-periodic extension to the whole $\R^d$,
one can rearrange \eqref{e1.12} as follows:
\begin{multline}
\label{e1.13}
\A(\xi,a,\mu) u(x)=\sum_{n\in\Z^d}\int_{\Omega} \bigl(a(x-y+n)u(x)-a(x-y+n)
e^{-i \langle \xi, x-y+n \rangle}u(y)\bigr)\mu(x,y) \,dy=
\\
= \int_{\R^d}a(x-y)\mu(x,y) \bigl(u(x)-e^{-i \langle \xi, x-y \rangle}u(y)\bigr) \,dy=
\\
= \int_{\R^d}a(x-y)\mu(x,y)e^{-i \langle \xi, x\rangle} \bigl(e^{i \langle \xi, x\rangle}u(x)
-e^{i \langle \xi, y \rangle}u(y)\bigr) \,dy,\ \ x\in\Omega,\ \ u\in L_{2}(\Omega).
\end{multline}
Then the quadratic form of the operator $\A(\xi,a,\mu)$ takes the form
\begin{multline}\label{e1.14}
(\A(\xi,a,\mu)u,u)=\int_{\Omega}dx\int_{\R^d} \,dy\,
a(x-y)\mu(x,y)e^{-i \langle \xi, x\rangle}\overline{u(x)} \bigl( e^{i \langle \xi, x \rangle}u(x)
-e^{i \langle \xi, y \rangle}u(y) \bigr) =
\\
=\int_{\Omega}dx\int_{\R^d}dy\, a(x-y)\mu(x,y) \bigl|e^{i \langle \xi, x \rangle}u(x)
-e^{i \langle \xi, y \rangle}u(y) \bigr|^{2}+J[u],
\ \ u\in L_{2}(\Omega),
\end{multline}
where the functional $J[u]$ is given by
\begin{equation}
\label{e1.15}
J[u]:=\int_{\Omega}dx\int_{\R^d}dy\, a(x-y)\mu(x,y)e^{-i \langle \xi, y \rangle}\overline{u(y)}
\bigl(e^{i \langle \xi, x \rangle}u(x)-e^{i \langle \xi, y\rangle}u(y) \bigr),\ \ u\in L_{2}(\Omega).
\end{equation}
Exchanging the variables $x$ and $y$ and considering conditions (\ref{e1.1}), (\ref{e1.2}) we conclude that
\begin{multline}\label{e1.16}
J[u]=-\int_{\R^d}dx \int_{\Omega}dy\, a(x-y)\mu(x,y)e^{-i \langle\xi, x\rangle}\overline{u(x)}
\bigl(e^{i \langle \xi, x\rangle}u(x)- e^{i \langle \xi, y \rangle}u(y)\bigr)=
\\
=-\sum_{n\in\Z^d}\int_{\Omega}dx\int_{\Omega}dy\, a(x-y+n)\mu(x,y)e^{-i \langle \xi, x+n \rangle}\overline{u(x)}
\bigl(e^{i \langle \xi, x+n \rangle}u(x)-e^{i \langle \xi, y\rangle}u(y)\bigr)=
\\
= -\sum_{n\in\Z^d}\int_{\Omega}dx\int_{\Omega}dy\, a(x-y+n)\mu(x,y)e^{-i \langle \xi, x \rangle}\overline{u(x)}
\bigl(e^{i \langle\xi, x\rangle} u(x)-e^{i \langle \xi, y-n \rangle}u(y)\bigr)=
\\
=-\int_{\Omega}dx\, \overline{u(x)}\int_{\Omega}dy\, \mu(x,y) \bigl(\widetilde{a}(0,x-y)u(x)-
\widetilde{a}(\xi,x-y)u(y)\bigr)=
\\
=-(\A(\xi,a,\mu)u,u),\ \ u\in L_{2}(\Omega).
\end{multline}
Now \eqref{e1.11} follows from (\ref{e1.14}) and (\ref{e1.16}).
\end{proof}

\subsection{Estimates of the quadratic form of operator $\A(\xi,a,\mu)$}
As above it is convenient to consider the case $\mu=\mu_{0}\equiv 1$ separately. Recalling
the notation $\A_{0}(\xi,a):=\A(\xi,a,\mu_{0})$, $\B_{0}(\xi,a):=\B(\xi,a,\mu_{0})$, from the relations
in   \eqref{e1.2} and \eqref{e1.11} we derive the estimates
\begin{equation}
\label{e1.17}
\mu_{-}(\A_{0}(\xi,a)u,u)\le(\A(\xi,a,\mu)u,u)\le\mu_{+}(\A_{0}(\xi,a)u,u),\
\ u\in L_{2}(\Omega),\ \ \xi\in\widetilde\Omega.
\end{equation}
The operators $\A_{0}(\xi,a)$, $\xi\in\widetilde\Omega$, are diagonalized by means of  the unitary discrete Fourier transform  $F:L_{2}(\Omega)\to \ell_{2}(\Z^d)$ defined as follows:
\begin{gather*}
Fu(n)=\int_{\Omega}u(x)e^{-2\pi i \langle n, x \rangle}dx,\ \ n\in\Z^d,\ \ u\in L_{2}(\Omega);
\\
F^{*}v(x)=\sum_{n\in\Z^d}v_{n}e^{2\pi i \langle n, x\rangle},\ \ x\in\Omega,\ \ v=\{v_{n}\}_{n\in\Z^d}\in\ell_{2}(\Z^d).
\end{gather*}
We have
\begin{equation}
\label{e1.18}
\A_{0}(\xi,a)=F^{*} \bigl[\hat a(0)-\hat a(2\pi n+\xi) \bigr] F,\ \ \hat a(k):=\int_{\R^d}a(x)e^{-i \langle k, x \rangle}dx,\ \ k\in\R^d;
\end{equation}
here $[f(n)]$, $f(n)=\hat a(0)-\hat a(2\pi n+\xi)$,  stands for the operator of multiplication by the function
$f(n)$ in the space  $l_2(\Z^d)$.
Thus, the symbol of the operator  $\A_{0}(\xi,a)$ is a sequence
$\{\hat A_{n}(\xi)\}_{n\in\Z^d}$ with
\begin{equation*}
\hat A_{n}(\xi)=\hat A(\xi+2\pi n)=\hat a(0)-\hat a(2\pi n+\xi)=\int_{\R^d} \bigl(1-\cos( \langle z, \xi+2\pi n \rangle) \bigr)a(z)\,dz, \ \ n\in\Z^d,\ \ \xi\in\widetilde\Omega.
\end{equation*}
Here we have used the fact that $a(z)$ is an even function and therefore  the integral $\int_{\R^d}\sin ( \langle z, \xi+2\pi n \rangle )a(z)\,dz$ is equal to zero.

\begin{lemma}
\label{lemma1.2}
Under conditions \eqref{e1.1}--\eqref{e1.3} the point $\lambda_{0}=0$ is a simple eigenvalue of the operator $\A(0,a,\mu)$.
Moreover, $\operatorname{Ker}\A(0,a,\mu)=\mathcal{L}\{\mathbf{1}_{\Omega}\}$.
\end{lemma}

\begin{proof}
Since the set of $z\in\mathbb R^d$ such that $a(z)>0$ has a positive Lebesgue  measure, the quantities
$\hat A(2\pi n)=\int_{\R^d}2\sin^{2}( \langle z, \pi n \rangle)a(z)\,dz$ are not equal to zero for
$n\in\Z^d\setminus\{0\}$; on the other hand, $\hat A(0)=0$. Therefore, $\operatorname{Ker}\A_{0}(0,a)=\mathcal{L}\{\mathbf{1}_{\Omega}\}$, and, by \eqref{e1.17},
$\operatorname{Ker}\A(0,a,\mu)=\mathcal{L}\{\mathbf{1}_{\Omega}\}$.
\end{proof}

In order to estimate  the quadratic form of the operator $\A_{0}(\xi,a)$ for all $\xi\in\widetilde\Omega$,
we carry out the detailed analysis of the symbol of the operator $\A_{0}(\xi,a)$, $\xi\in\widetilde\Omega$.
Under condition \eqref{e1.1} the function
\begin{equation}
\label{e1.17a}
\hat A(y)=\int_{\R^d} \bigl(1-\cos(\langle z, y\rangle) \bigr)a(z)\,dz=2\int_{\R^d}\sin^{2}\left(\frac{\langle z, y\rangle}{2}\right)a(z)\,dz,\ \ y\in\R^d,
\end{equation}
depends continuously on
$y\in\R^d$ and,  by the Riemann--Lebesgue lemma, converges to
\hbox{$\|a\|_{L_1}>0$} as $|y|\to\infty$.
Furthermore, it is easy to check that  $\hat A(y)\not=0$ if $y\not=0$.
Consequently, the estimate
\begin{equation}
\label{e1.19}
\min_{|y| \ge r}\hat A(y)=:\mathcal{A}_{r}(a)>0,\ \ r>0,
\end{equation}
holds. Under condition \eqref{e1.4} the function
\begin{equation}
\label{e1.18a}
\int_{\R^d}a(z) \langle z, y \rangle^{2}\, dz=:M_{a}(y)
\end{equation}
is continuous in  $y\in\R^d$ and not equal to zero if $y\not=0$. This yields the following inequality:
\begin{equation}
\label{e1.20}
\min_{|\theta|=1}M_{a}(\theta)=:\mathcal{M}(a)>0.
\end{equation}
Below, in Section  \ref{Append2}, under conditions \eqref{e1.1}--\eqref{e1.4} and an additional assumption $a\in L_{2}(\R^d)$,
we will obtain explicit expressions for $\mathcal{A}_{r}(a)$, $r>0$, and $\mathcal{M}(a)$ in terms of
$\|a\|_{L_1}$, $\|a\|_{L_2}$ and $M_{k}(a)$ with $k=1,2,3$.

\begin{lemma}
\label{lemma1.3}
Let conditions \eqref{e1.1}--\eqref{e1.4} be fulfilled. Then
\begin{equation}
\label{e1.21}
\hat A(\xi+2\pi n)\ge C(a)|\xi|^{2},\ \ \xi\in\widetilde\Omega,\ \ n\in\Z^d;
\end{equation}
\begin{equation}
\label{e1.22}
\hat A(\xi+2\pi n)\ge \mathcal{A}_{\pi}(a),\ \ \xi\in\widetilde\Omega,\ \ n\in\Z^d\setminus\{0\},
\end{equation}
where
\begin{equation}
\label{C(a)}
C(a):=\min \Bigl\{\frac{1}{4}\mathcal{M}(a),\mathcal{A}_{r(a)}(a)\pi^{-2}d^{-1},\mathcal{A}_{\pi}(a)\pi^{-2}d^{-1}
\Bigr\}>0,
\end{equation}
the quantities  $\mathcal{A}_{r}(a)$, $r>0$, and $\mathcal{M}(a)$ are defined in \eqref{e1.19}
and \eqref{e1.20}, respectively, and $r(a):=\frac{3}{2}\mathcal{M}(a)M_{3}^{-1}(a)$.
\end{lemma}

\begin{proof}
If conditions  \eqref{e1.1}, \eqref{e1.4} hold, then the function
$\hat A(y)$, $y\in\R^d$, is three times continuously differentiable, and
\begin{equation}
\label{e1.23}
\hat A(0)=0,\ \ \nabla\hat A(0)=0,\ \ \langle (H\hat A)(0) y,y\rangle :=\sum_{i,j=1}^{d}\partial_{i}\partial_{j}\hat
A(0)y_{i}y_{j}=\int_{\R^d} \langle z, y \rangle^{2} a(z)\,dz,\ \ y\in\R^d;
\end{equation}
\begin{equation}
\label{e1.24}
\partial_{i}\partial_{j}\partial_{k}\hat
A(y)=-\int_{\R^d}z_{i}z_{j}z_{k}\sin( \langle z, y \rangle) a(z)\,dz,\ \ y\in\R^d.
\end{equation}
Combining condition  (\ref{e1.4}) with relations (\ref{e1.20}) and
(\ref{e1.23}) we obtain the estimates
\begin{equation}\label{e1.25}
\mathcal{M}(a)|y|^{2}\le\langle (H\hat A)(0) y,y\rangle\le
M_{2}(a)|y|^{2},\ \ y\in\R^d.
\end{equation}
The following inequality is a consequence of  \eqref{e1.4} and \eqref{e1.24}:
\begin{equation}\label{e1.26}
\left|\sum_{i=1}^{d}\sum_{j=1}^{d}\sum_{k=1}^{d}y_{i}y_{j}y_{k}\partial_{i}\partial_{j}\partial_{k}\hat
A(y_{0})\right|\le M_{3}(a)|y|^{3},\ \ y,y_{0}\in\R^d.
\end{equation}

By the Hadamard lemma (see Lemma  \ref{Adamar} below)  and due to relations  \eqref{e1.23} and \eqref{e1.26},  we have
\begin{equation}
\label{e1.27}
\hat A(y)=\frac{1}{2} \langle (H\hat A)(0)y,y \rangle +R(y),\ \
|R(y)|\le\frac{1}{6}M_{3}(a)|y|^{3},\ \ y\in\R^d.
\end{equation}
The second relation in \eqref{e1.27} implies that $|R(y)|\le\frac{1}{4}\mathcal{M}(a)|y|^{2}$ for
$|y|\le r(a):=\frac{3}{2}\mathcal{M}(a)M_{3}^{-1}(a)$.
Therefore, considering \eqref{e1.25} and \eqref{e1.27} we obtain the following lower bound:
\begin{equation}\label{e1.28}
\hat A(y)\ge\frac{1}{4}\mathcal{M}(a)|y|^{2},\ \ |y|\le r(a).
\end{equation}
Observe that
\begin{equation}\label{e1.29}
\begin{gathered}
\xi\in\widetilde\Omega\Longrightarrow |\xi|\le\pi\sqrt{d};
\\
\xi\in\widetilde\Omega,\ \ n\in\Z^d \Longrightarrow |\xi+2\pi n|\ge |\xi|;
\\
\xi\in\widetilde\Omega,\ \ n\in\Z^d\setminus\{0\}\Longrightarrow |\xi+2\pi n|\ge\pi.
\end{gathered}
\end{equation}
Combining these relations with  \eqref{e1.19} and  \eqref{e1.28} yields
\begin{equation*}
\begin{array}{l}
\hat A(\xi+2\pi n)\ge \frac{1}{4}\mathcal{M}(a)|\xi|^{2},\ \ n=0,\ \ \xi\in\widetilde\Omega,\ \ |\xi|\le r(a);
\\
\hat A(\xi+2\pi n)\ge\mathcal{A}_{r(a)}(a)\ge\mathcal{A}_{r(a)}(a)\pi^{-2}d^{-1}|\xi|^{2},\
\ n=0,\ \ \xi\in\widetilde\Omega,\ \ |\xi|\ge r(a);
\\
\hat A(\xi+2\pi n)\ge\mathcal{A}_{\pi}(a)\ge\mathcal{A}_{\pi}(a)\pi^{-2}d^{-1}|\xi|^{2},\
\ n\in\Z^d\setminus\{0\},\ \ \xi\in\widetilde\Omega.
\end{array}
\end{equation*}
Finally, we conclude that inequalities \eqref{e1.21} and  \eqref{e1.22} hold true.
\end{proof}

The following statement is a consequence of relations \eqref{e1.17}, \eqref{e1.18}
and Lemmata \ref{lemma1.2},  \ref{lemma1.3}:

\begin{proposition}
\label{prop1.4}
Assume that conditions  \eqref{e1.1}--\eqref{e1.4} hold.  Then we have
\begin{equation}\label{e1.30}
(\A(\xi,a,\mu)u,u)\ge \mu_{-}C(a)|\xi|^{2}\|u\|_{L_2(\Omega)}^{2},\ \ u\in
L_{2}(\Omega),\ \ \xi\in\widetilde\Omega;
\end{equation}
\begin{equation}\label{e1.31}
(\A(\xi,a,\mu)u,u)\ge \mu_{-}\mathcal{A}_{\pi}(a)\|u\|^{2}_{L_2(\Omega)},\ \
u\in L_{2}(\Omega)\ominus\mathcal{L}\{\mathbf{1}_{\Omega}\},\ \
\xi\in\widetilde\Omega.
\end{equation}
\end{proposition}
We will also need
\begin{proposition}
\label{prop1.5}
Let conditions  \eqref{e1.1}--\eqref{e1.4} be fulfilled. Then the following estimate holds:
\begin{equation}
\label{e1.32}
\|\A(\xi,a,\mu)-\A(\eta,a,\mu)\|\le\mu_{+}2^{d}M_{1}(a)|\xi-\eta|,\ \ \xi,\eta\in\widetilde\Omega.
\end{equation}
\end{proposition}
\begin{proof}
Since $\A(\xi,a,\mu)-\A(\eta,a,\mu)$ is an integral operator in $L_{2}(\Omega)$ with the kernel
\hbox{$\bigl(\widetilde a(\eta,x-y)-\widetilde a(\xi,x-y)\bigr)\mu(x,y)$}, $x,y\in\Omega$, then, by the Schur test
(see Lemma \ref{lemma4.1}), we obtain
\begin{multline*}
\|\A(\xi,a,\mu)-\A(\eta,a,\mu)\|\le\mu_{+}\intop_{[-1,1]^{d}} \bigl|\widetilde a(\eta,z)-\widetilde a(\xi,z)\bigr|\,dz
=\mu_{+}2^{d}\int_{\Omega} \bigl|\widetilde a(\eta,z)-\widetilde a(\xi,z)\bigr| \, dz\le
\\
\le\mu_{+}2^{d}\int_{\R^d}a(z) \bigl|e^{-i \langle \eta, z\rangle}-e^{-i \langle \xi, z\rangle} \bigr|\,dz
=\mu_{+}2^{d} \int_{\R^d}a(z) \bigl| e^{i \langle \xi-\eta, z \rangle}-1 \bigr|\,dz\le
\\
\le\mu_{+}2^{d} |\xi-\eta| \int_{\R^d}a(z) |z| \,dz=\mu_{+}2^{d}M_{1}(a)|\xi-\eta|,\
\ \xi,\eta\in\widetilde\Omega.
\end{multline*}
\vskip-3mm \end{proof}

\section{Threshold characteristics of nonlocal Schr\"odinger operator \\  in the vicinity of the spectrum bottom}

\subsection{The spectral edge of the operator $\A(\xi,a,\mu)$\label{subsec_21}}
According to Lemma \ref{lemma1.2}, under conditions  (\ref{e1.1})--(\ref{e1.3}) the lower edge of the spectrum of operator
$\A(0,a,\mu)$ which is defined as $\min\,\{\lambda\in\mathbb R\,:\,\lambda\in \sigma(\A(0,a,\mu))\}$,  is a simple eigenvalue  $\lambda_{0}=0$.
Letting $d_{0}(a,\mu)=\mathrm{dist}(\lambda_0\,,\, \sigma(\A(0,a,\mu))\setminus\{\lambda_0\})$ and assuming that
conditions  \eqref{e1.1}--\eqref{e1.4} are fulfilled,    we derive from estimate
 (\ref{e1.31}) the following inequality:
\begin{equation}
\label{d0}
d_{0}(a,\mu)\ge\mu_{-}\mathcal{A}_{\pi}(a).
\end{equation}
For brevity we use the short notation $d_0:=d_{0}(a,\mu)$. By the perturbation theory arguments,  we deduce from estimate
\eqref{e1.32} that  for all $\xi\in\widetilde\Omega$ such that
$|\xi|\le3^{-1}2^{-d}M_{1}(a)^{-1}\mu_{+}^{-1}\mu_{-}\mathcal{A}_{\pi}(a)$
the relations
\begin{equation*}
\operatorname{rank}E_{\A(\xi,a,\mu)}[0,d_{0}/3]=1,\ \
\sigma(\A(\xi,a,\mu))\cap(d_{0}/3;2d_{0}/3)=\varnothing
\end{equation*}
hold.
Denote $\delta_{0}(a,\mu)=\min\{1,3^{-1}2^{-d}M_{1}(a)^{-1}\mu_{+}^{-1}\mu_{-}\mathcal{A}_{\pi}(a)\}$. Observe that
$\delta_{0}(a,\mu)>0$.
Combining the above statements we arrive at
\begin{proposition}
\label{prop2.1}
Let conditions  \eqref{e1.1}--\eqref{e1.4} be satisfied. Then for all $|\xi|$ with
$|\xi|\le\delta_{0}(a,\mu)$ the spectrum of the operator  $\A(\xi) = \A(\xi,a,\mu)$ in the
interval $[0,d_{0}/3]$ consists of just one simple eigenvalue,  while the interval  $(d_{0}/3,2d_{0}/3)$
belongs to the resolvent set  of $\A(\xi,a,\mu)$.
\end{proposition}

Under conditions \eqref{e1.1}--\eqref{e1.4} the operator-function $\A(\xi,a,\mu)$ is three times continuously differentiable
in the variable $\xi$ with respect to the operator norm  of $\mathcal{B}(L_{2}(\Omega))$, moreover, the following relations hold
\begin{multline}
\label{e2.1}
\partial^{\alpha}\A(\xi,a,\mu)u(x)=\int_{\Omega}\widetilde
a_{\alpha}(\xi,x-y)\mu(x,y)u(y)\, dy,\ \ x\in\Omega,\ \ u\in L_{2}(\Omega);
\\
\widetilde a_{\alpha}(\xi,z)=(-1)(-i)^{|\alpha|}\sum_{n\in\Z^d}(z+n)^{\alpha}a(z+n)e^{-i \langle \xi, z+n \rangle},\
\ \alpha\in\Z^d_{+},\ \ |\alpha|\le 3.
\end{multline}
By the Hadamard lemma (see Lemma  \ref{Adamar} below), the operator-function  $\A(\cdot,a,\mu)$ admits an expansion
\begin{multline}\label{e2.2}
\A(\xi,a,\mu)=\A(0,a,\mu)+\sum_{i=1}^{d}\partial_{i}\A(0,a,\mu)\xi_{i}+\\+
\frac{1}{2}\sum_{i,j=1}^{d}\partial_{i}\partial_{j}\A(0,a,\mu)\xi_{i}\xi_{j}+\mathbb{K}(\xi,a,\mu),\
\ |\xi|\le\delta_{0}(a,\mu),
\end{multline}
with
$$
\mathbb{K}(\xi,a,\mu) =
\sum_{i=1}^{d}\sum_{j=1}^{d}\sum_{k=1}^{d} \xi_{i} \xi_{j} \xi_{k}\intop_{0}^{1}ds_{1}s_{1}^{2}\intop_{0}^{1}ds_{2}\,s_{2}\intop_{0}^{1}ds_{3} \,\partial_{i}\partial_{j}\partial_{k} {\mathbb A}(s_{1}s_{2}s_{3}\xi).
$$
Applying the Schur test one can easily show that
\begin{equation}\label{e2.3}
\Big\|\sum_{i=1}^{d}\partial_{i}\A(0,a,\mu)\xi_{i}\Big\|\le|\xi|\mu_{+} M_{1}(a),\
\ |\xi|\le\delta_{0}(a,\mu);
\end{equation}
\begin{equation}\label{e2.4}
\Big\|\frac{1}{2}\sum_{i,j=1}^{d}\partial_{i}\partial_{j}\A(0,a,\mu)\xi_{i}\xi_{j}\Big\|\le|\xi|^{2}\frac{1}{2}\mu_{+} M_{2}(a),\
\ |\xi|\le\delta_{0}(a,\mu);
\end{equation}
\begin{equation}\label{e2.5}
\|\mathbb{K}(\xi,a,\mu)\|\le|\xi|^{3}\frac{1}{6}\mu_{+} M_{3}(a),\
\ |\xi|\le\delta_{0}(a,\mu).
\end{equation}
We introduce some more notation: $F(\xi)$ denotes the spectral projector of the operator $\A(\xi,a,\mu)$ that corresponds
to the interval $[0,d_{0}/3]$, the symbol $\mathfrak{N}$ stands for the kernel $\operatorname{Ker}\A(0,a,\mu)=\mathcal{L}\{\mathbf{1}_{\Omega}\}$.  We then denote by  $P$ the orthogonal projector on $\mathfrak{N}$ and set   $P^{\perp}:=I-P$.
Let $\Gamma$ be a contour on the complex plane that is equidistant to the interval $[0,d_{0}/3]$ and passes through the middle
point of the interval $(d_{0}/3,2d_{0}/3)$.
By the Riesz formula that reads
\begin{equation}
\label{e2.6}
F(\xi)=\frac{-1}{2\pi i}\ointop_{\Gamma}(\A(\xi,a,\mu)-\zeta I)^{-1}d\zeta,\ \ |\xi|\le\delta_{0}(a,\mu),
\end{equation}
the operator-functions $F(\xi)$ and  $\A(\xi,a,\mu)F(\xi)$, $|\xi|\le\delta_{0}(a,\mu)$, are three times continuously differentiable. Therefore, we have
\begin{equation}
\label{e2.7}
F(\xi)=P+\sum_{i=1}^{d}F_{i}\xi_{i}+\frac{1}{2}\sum_{i,j=1}^{d}F_{ij}\xi_{i}\xi_{j}+O(|\xi|^{3}),\
\ |\xi|\to 0;
\end{equation}
\begin{equation}\label{e2.8}
\A(\xi,a,\mu)F(\xi)=G_{0}+\sum_{i=1}^{d}G_{i}\xi_{i}+\frac{1}{2}\sum_{i,j=1}^{d}G_{ij}\xi_{i}\xi_{j}+O(|\xi|^{3}),\
\ |\xi|\to 0.
\end{equation}
The main goal of this section is to determine the coefficients of expansions  \eqref{e2.7}, \eqref{e2.8} and to estimate the remainders.

\subsection{Computing the coefficients of the expansions of $F(\xi)$ and $\A(\xi)F(\xi)$}
\label{subsec_22}
From \eqref{e2.2}, \eqref{e2.5} and  \eqref{e2.7} we derive the relations
\begin{equation}\label{e2.9}
G_{0}=\A(0,a,\mu)P=0;
\end{equation}
\begin{equation}\label{e2.10}
G_{i}=\partial_{i}\A(0,a,\mu)P+\A(0,a,\mu)F_{i},\ \ i=1,\dots,d;
\end{equation}
\begin{equation}\label{e2.11}
G_{ij}=\partial_{i}\partial_{j}\A(0,a,\mu)P+\partial_{i}\A(0,a,\mu)F_{j}+\partial_{j}\A(0,a,\mu)F_{i}+\A(0,a,\mu)F_{ij},\
\ i,j=1,\dots,d.
\end{equation}
Since the operators $\A(\xi,a,\mu)F(\xi)$ are self-adjoint, the second relation implies that the operators $G_{i}$, $i=1,\dots,d$, are
also self-adjoint. Consequently, taking into account the inequality  $(\A(\xi,a,\mu)F(\xi)u,u)\ge 0$,
$u\in L_{2}(\Omega)$, $|\xi|\le\delta_{0}(a,\mu)$,  the equality $\A(0,a,\mu)F(0)=0$
and the definition of $G_i$ in \eqref{e2.10}, we conclude that
\begin{equation}\label{e2.12}
G_{i}=\partial_{i}\A(0,a,\mu)P+\A(0,a,\mu)F_{i}=0,\ \ i=1,\dots,d.
\end{equation}
Combining \eqref{e2.8}, \eqref{e2.9} and  \eqref{e2.12} yields the following asymptotic representation:
\begin{equation}
\label{e2.13}
\A(\xi,a,\mu)F(\xi)=\frac{1}{2}\sum_{i,j=1}^{d}G_{ij}\xi_{i}\xi_{j}+O(|\xi|^{3}),\
\ |\xi|\to 0.
\end{equation}
Since $F(\xi)$ is a spectral projector, we have $F^{2}(\xi)=F(\xi)$ and thus,  recalling that $P$ is the orthogonal projection on $\mathcal{L}\{\mathbf{1}_{\Omega}\}$, we obtain
\begin{equation}
\label{e2.14}
F_{i}P+PF_{i}=F_{i},\ \ i=1,\dots,d.
\end{equation}
The latter relation implies that  $P^{\perp}F_{i}P^{\perp}=0$ and  $PF_{i}P=2PF_{i}P$.
Therefore,  $PF_{i}P=0$ and, by \eqref{e2.14},
\begin{equation}
\label{e2.15}
F_{i}=P^{\perp}F_{i}P+PF_{i}P^{\perp},\ \ i=1,\dots,d.
\end{equation}
From the latter inequality and  \eqref{e2.12} we deduce that
\begin{equation*}
\partial_{i}\A(0,a,\mu)P=-\A(0,a,\mu)F_{i}=-P^{\perp}\A(0,a,\mu)P^{\perp}F_{i}=-P^{\perp}\A(0,a,\mu)P^{\perp}F_{i}P,\
\ i=1,\dots,d.
\end{equation*}
Finally, we arrive at the following relations:
\begin{equation}\label{e2.16}
\partial_{i}\A(0,a,\mu)P=P^{\perp}\partial_{i}\A(0,a,\mu)P,
\end{equation}
\begin{equation}\label{e2.17}
P^{\perp}F_{i}P=-P^{\perp}\A(0,a,\mu)^{-1}P^{\perp}\partial_{i}\A(0,a,\mu)P,\
\ i=1,\dots,d;
\end{equation}
here the notation $\A(0,a,\mu)^{-1}$ is used for  the operator that is inverse to the operator
$\A(0,a,\mu)\vert_{{\mathfrak N}^\perp}\,:\,{\mathfrak N}^\perp \mapsto {\mathfrak N}^\perp$. Notice
that this operator is well defined.

Next, considering \eqref{e2.13}, the asymptotic formula $F(\xi)=P+O(|\xi|)$ and the properties of the operator
$\A(\xi)F(\xi)=(\A(\xi)F(\xi))F(\xi)=F(\xi)(\A(\xi)F(\xi))$ one has $G_{ij}=PG_{ij}P$, $i,j=1,\dots,d$.
Consequently, in view of \eqref{e2.15}, relation \eqref{e2.11} can be rearranged as follows:
\begin{equation}\label{e2.18}
G_{ij}=P\partial_{i}\partial_{j}\A(0,a,\mu)P+P\partial_{i}\A(0,a,\mu)P^{\perp}F_{j}P+P\partial_{j}\A(0,a,\mu)P^{\perp}F_{i}P,\
\ i,j=1,\dots,d.
\end{equation}
Substituting the right-hand side of  \eqref{e2.17} for  $P^{\perp}F_{j}P$ in \eqref{e2.18} yields
\begin{multline}\label{e2.19}
G_{ij}=P\partial_{i}\partial_{j}\A(0,a,\mu)P-P\partial_{i}\A(0,a,\mu)P^{\perp}\A(0,a,\mu)^{-1}P^{\perp}\partial_{j}\A(0,a,\mu)P-\\
-P\partial_{j}\A(0,a,\mu)P^{\perp}\A(0,a,\mu)^{-1}P^{\perp}\partial_{i}\A(0,a,\mu)P,\
\ i,j=1,\dots,d.
\end{multline}
Since $P=(\cdot,\mathbf{1}_{\Omega})\mathbf{1}_{\Omega}$ is a first rank operator,  $G_{ij}$ admits a representation
$G_{ij}=g_{ij}P$ with the coefficients
$g_{ij}\in\mathbb{C}$, $i,j=1,\dots,d$, that can be computed with the help of formula \eqref{e2.19}.
Indeed, from the relation $P=(\cdot,\mathbf{1}_{\Omega})\mathbf{1}_{\Omega}$  we obtain
\begin{equation}\label{e2.20}
\partial_{j}\A(0,a,\mu)P=\boldsymbol{\mathtt{i}}(\cdot,\mathbf{1}_{\Omega})w_{j}\ \ \
\text{with}\ \ {\boldsymbol{\mathtt{i}}}w_{j}=\partial_{j}\A(0,a,\mu)\mathbf{1}_{\Omega},\
\ j=1,\dots,d.
\end{equation}
Due to \eqref{e2.1} the following equality holds:
\begin{multline}
\label{e2.21}
w_{j}(x)=\overline{w_{j}(x)}=\intop_{\Omega}\sum_{n\in\Z^d}(x_{j}-y_{j}+n_{j})a(x-y+n)\mu(x,y)\,dy=
\\
= \intop_{\R^d}(x_{j}-y_{j})a(x-y)\mu(x,y)\,dy,\ \ x\in\Omega,\ \ j=1,\dots,d.
\end{multline}
From  \eqref{e2.16} and \eqref{e2.20} it follows that, for all $j=1,\dots,d$,
\begin{equation}
\label{e2.22}
P^{\perp}\A(0,a,\mu)^{-1}P^{\perp}\partial_{j}\A(0,a,\mu)P={\boldsymbol{\mathtt{i}}}(\cdot,\mathbf{1}_{\Omega})v_{j},
\ \ v_{j}=P^{\perp}\A(0,a,\mu)^{-1}P^{\perp}w_{j}.
\end{equation}
Notice that the functions $v_{j}=\overline{v_{j}}\in L_{2}(\Omega)$, $j=1,\dots,d$, satisfy  cell problems
on $\Omega$ that read
\begin{equation*}
\left\{
\begin{array}{l}
\int\limits_{\Omega}\widetilde
a(0,x-y)\mu(x,y)(v_{j}(x)-v_{j}(y))\,dy=w_{j}(x),\ \ x\in\Omega,\\[4mm]
\int\limits_{\Omega}v_{j}(x)\,dx=0,
\end{array}
\right. j=1,\dots,d.
\end{equation*}
Identifying the functions $v_{j}\in L_{2}(\Omega)$, $j=1,\dots,d$, with their periodic extensions one can rewrite
these cell problems as follows:
\begin{multline}\label{e2.23}
\left\{
\begin{array}{l}
\intop\limits_{\R^d}
a(x-y)\mu(x,y)(v_{j}(x)-v_{j}(y))\,dy=\intop\limits_{\R^d}
a(x-y)\mu(x,y)(x_{j}-y_{j}) \,dy,\ \ x\in\Omega,
\\[4mm]
\intop\limits_{\Omega}v_{j}(x)\,dx=0.
\end{array}
\right.
\end{multline}
Making the rearrangements similar to those in   \eqref{e1.6} we conclude that the problems in \eqref{e2.23} have a unique solution.
Then, by \eqref{e2.20} and \eqref{e2.22}, we have
\begin{equation}\label{e2.24}
P\partial_{i}\A(0,a,\mu)P^{\perp}\A(0,a,\mu)^{-1}P^{\perp}\partial_{j}\A(0,a,\mu)P=(v_{j},w_{i})P,\
\ i,j=1,\dots,d,
\end{equation}
where the functions $w_{j}\in L_{2}(\Omega)$, $j=1,\dots,d$ are defined in \eqref{e2.21}, and the functions
 $v_{j}\in L_{2}(\Omega)$, $j=1,\dots,d$, satisfy auxiliary problems \eqref{e2.23} on the periodicity cell $\Omega$.
In the same way one can derive the relations
\begin{equation}
\label{e2.25}
P\partial_{i}\partial_{j}\A(0,a,\mu)P=(w_{ij},\mathbf{1}_{\Omega})P,\
\ i,j=1,\dots,d.
\end{equation}
Here
$w_{ij}=\overline{w_{ij}}=\partial_{i}\partial_{j}\A(0,a,\mu)\mathbf{1}_{\Omega}\in L_{2}(\Omega)$, that is
\begin{multline}\label{e2.26}
w_{ij}(x)=\intop_{\Omega}\sum_{n\in\Z^d}(x_{i}-y_{i}+n_{i})(x_{j}-y_{j}+n_{j})a(x-y+n)\mu(x,y)\,dy=
\\
= \intop_{\R^d}(x_{i}-y_{i})(x_{j}-y_{j})a(x-y)\mu(x,y)\, dy,\ \
x\in\Omega,\ \ i,j=1,\dots,d.
\end{multline}
Finally, combining the above relations   \eqref{e2.19}, \eqref{e2.21} and  \eqref{e2.24}--\eqref{e2.26}
we conclude that
\begin{multline}\label{e2.27}
g_{ij}=(w_{ij},\mathbf{1}_{\Omega})-(v_{j},w_{i})-(v_{i},w_{j})=
\\
= \intop_{\Omega}dx\intop_{\R^d} dy((x_{i}-y_{i})(x_{j}-y_{j})-v_{j}(x)(x_{i}-y_{i})-v_{i}(x)(x_{j}-y_{j}))a(x-y)\mu(x,y),
\\
i,j=1,\dots,d.
\end{multline}
Here the functions $v_{j}(x)$, $x\in\Omega$, $j=1,\dots,d$, are defined as solutions of problems \eqref{e2.23}.

It should be noted that there are other methods of computing the coefficients of Taylor expansions for the functions
   $F(\xi)$ and $\A(\xi)F(\xi)$.  Two of them are described in Section \ref{subsec4_2_other}.

\subsection{Approximation of the operators $F(\xi)$ and $\A(\xi)F(\xi)$}
The precision of the above constructed approximations of the operator-functions $F(\xi)$ and $\A(\xi,a,\mu)F(\xi)$
for $|\xi|\le\delta_{0}(a,\mu)$ is estimated in the next statement.

\begin{theorem}
\label{teor2.1}
Let conditions \eqref{e1.1}--\eqref{e1.4} be fulfilled.
Then the following estimates hold:
\begin{equation}\label{e2.28}
\|F(\xi)-P\|\le C_{1}(a,\mu)|\xi|,\ \ |\xi|\le\delta_{0}(a,\mu);
\end{equation}
\begin{equation}\label{e2.29}
\Big\|\A(\xi,a,\mu)F(\xi)-\frac{1}{2}\sum_{i,j=1}^{d}g_{ij}\xi_{i}\xi_{j}P\Big\|\le
C_{2}(a,\mu)|\xi|^{3},\ \ |\xi|\le\delta_{0}(a,\mu).
\end{equation}
The constants $C_{1}(a,\mu)$ and $C_{2}(a,\mu)$ defined in \eqref{C1} and \eqref{C2}, respectively,
depend only on $d$, $\mu_-$, $\mu_+$, $M_1(a)$, $M_2(a)$, $M_3(a)$, ${\mathcal A}_\pi(a)$.
\end{theorem}

\begin{proof}
By the Riesz formula \eqref{e2.6}, the difference  $F(\xi)-P$ admits the represen\-tation
\begin{equation}\label{e2.30}
F(\xi)-P=\frac{-1}{2\pi i}\oint_{\Gamma} \left((\A(\xi,a,\mu)-\zeta
I)^{-1}-(\A(0,a,\mu)-\zeta I)^{-1}\right) \, d\zeta,\ \
|\xi|\le\delta_{0}(a,\mu).
\end{equation}
The term  $\left((\A(\xi,a,\mu)-\zeta I)^{-1}-(\A(0,a,\mu)-\zeta I)^{-1}\right)$ satisfies the resolvent identity
that reads
\begin{multline}\label{e2.31}
(\A(\xi,a,\mu)-\zeta I)^{-1}-(\A(0,a,\mu)-\zeta I)^{-1}=
\\
=(\A(\xi,a,\mu)-\zeta I)^{-1} \left(\A(0,a,\mu)-\A(\xi,a,\mu)\right)(\A(0,a,\mu)-\zeta I)^{-1},
\\
|\xi|\le\delta_{0}(a,\mu),\ \ \zeta\in\Gamma.
\end{multline}
It remains to notice that the length of the contour $\Gamma$ is equal to $\frac{\pi+2}{3}d_{0}$ and that
for all $\zeta\in\Gamma$ both resolvents satisfy the estimates
\begin{equation}\label{e2.32}
\|(\A(\xi,a,\mu)-\zeta I)^{-1}\|\le 6d_{0}^{-1},\ \
\|(\A(0,a,\mu)-\zeta I)^{-1}\|\le 6d_{0}^{-1}.
\end{equation}
Now, as a consequence of  \eqref{e1.32} and \eqref{e2.30}--\eqref{e2.32}, we obtain \eqref{e2.28}
with
\begin{equation}
\label{C1}
C_{1}(a,\mu):=6(\pi+2)\pi^{-1}2^{d}d_{0}^{-1}\mu_{+}M_{1}(a).
\end{equation}

Denote
\begin{align*}
R(\xi,\zeta)&:=(\A(\xi,a,\mu)-\zeta I)^{-1},\ \
|\xi|\le\delta_{0}(a,\mu),\ \ \zeta\in\Gamma;
\\
R_{0}(\zeta)&:=R(0,\zeta),\ \ \zeta\in\Gamma;
\\
\Delta\A(\xi)&:=\A(\xi,a,\mu)-\A(0,a,\mu),\ \
|\xi|\le\delta_{0}(a,\mu).
\end{align*}
Iterating \eqref{e2.31} we derive the relation
\begin{multline}\label{e2.33}
R(\xi,\zeta)=R_{0}(\zeta)-R_{0}(\zeta)\Delta\A(\xi)R_{0}(\zeta)+\\+
R_{0}(\zeta)\Delta\A(\xi)R_{0}(\zeta)\Delta\A(\xi)R_{0}(\zeta)-\\-
R(\xi,\zeta)\Delta\A(\xi)R_{0}(\zeta)\Delta\A(\xi)R_{0}(\zeta)\Delta\A(\xi)R_{0}(\zeta),\
\ |\xi|\le\delta_{0}(a,\mu),\ \ \zeta\in\Gamma.
\end{multline}
From \eqref{e2.2} and \eqref{e2.33} it follows that
\begin{multline}\label{e2.33+}
R(\xi,\zeta)=R_{0}(\zeta)-R_{0}(\zeta)\sum_{i=1}^{d}\partial_{i}\A(0,a,\mu)\xi_{i}R_{0}(\zeta)-
R_{0}(\zeta)\frac{1}{2}\sum_{i,j =1}^{d}\partial_{i}\partial_{j}\A(0,a,\mu)\xi_{i}\xi_{j}R_{0}(\zeta)+
\\
+R_{0}(\zeta)\sum_{i=1}^{d}\partial_{i}\A(0,a,\mu)\xi_{i}R_{0}(\zeta)\sum_{j=1}^{d}\partial_{j}\A(0,a,\mu)\xi_{j}
R_{0}(\zeta)- R_{0}(\zeta)\mathbb{K}(\xi,a,\mu)R_{0}(\zeta)+
\\
+ R_{0}(\zeta)\sum_{i=1}^{d}\partial_{i}\A(0,a,\mu)\xi_{i}R_{0}(\zeta)\Delta_{2}\A(\xi)R_{0}(\zeta)+
R_{0}(\zeta)\Delta_{2}\A(\xi) R_{0}(\zeta)\Delta\A(\xi)R_{0}(\zeta)-
\\
-R(\xi,\zeta)\Delta\A(\xi)R_{0}(\zeta)\Delta\A(\xi)R_{0}(\zeta)\Delta\A(\xi)R_{0}(\zeta),\
\ |\xi|\le\delta_{0}(a,\mu),\ \ \zeta\in\Gamma.
\end{multline}
Here
$$
\Delta_{2}\A(\xi) := \frac{1}{2}\sum_{i,j =1}^{d}\partial_{i}\partial_{j}\A(0,a,\mu)\xi_{i}\xi_{j} + \mathbb{K}(\xi,a,\mu).
$$
Taking into account \eqref{e2.33+} and the relation
\begin{equation}\label{e2.34}
\A(\xi,a,\mu)F(\xi)=\frac{-1}{2\pi
i}\oint_{\Gamma}(\A(\xi,a,\mu)-\zeta I)^{-1}\zeta \,d\zeta,\ \
|\xi|\le\delta_{0}(a,\mu),
\end{equation}
we obtain the following representations for the coefficients of the expansion in \eqref{e2.8}:
\begin{equation}\label{dop2.1}
G_{0}=\frac{-1}{2\pi i}\oint_{\Gamma}R_{0}(\zeta)\zeta\, d\zeta;
\end{equation}
\begin{equation}\label{dop2.2}
G_{i}=\frac{1}{2\pi
i}\oint_{\Gamma}R_{0}(\zeta)\partial_{i}\A(0,a,\mu)R_{0}(\zeta)\zeta\, d\zeta,\ \ i=1,\dots,d;
\end{equation}
\begin{multline}\label{dop2.3}
G_{ij}=\frac{1}{2\pi i}\oint_{\Gamma}R_{0}(\zeta)\partial_{i}
\partial_{j}\A(0,a,\mu)R_{0}(\zeta)\zeta\,d\zeta-
\\
-\frac{1}{2\pi i}\oint_{\Gamma}R_{0}(\zeta)\Big(\partial_{i}\A(0,a,\mu)R_{0}(\zeta)\partial_{j}\A(0,a,\mu)+\partial_{j}\A(0,a,\mu)R_{0}(\zeta)\partial_{i}\A(0,a,\mu)\Big) R_{0}(\zeta)\zeta \,d\zeta,\\ i,j=1,\dots,d.
\end{multline}
Due to \eqref{e2.9} and \eqref{e2.12} we have
$G_{0}=0$, $G_{i}=0$, $i=1,\dots,d$; then the representation
\eqref{e2.8} takes the form \eqref{e2.13} with $G_{ij}=g_{ij}P$
and $g_{ij}\in\mathbb{C}$  introduced in \eqref{e2.27}.
In view of  \eqref{e2.33+} and  \eqref{e2.34} the discrepancy in the expansion
\eqref{e2.13} admits the following representation:
\begin{multline}\label{e2.34+}
\A(\xi,a,\mu)F(\xi)-\frac{1}{2}\sum_{i,j=1}^{d}G_{ij}\xi_{i}\xi_{j}=\frac{1}{2\pi i}
\oint_{\Gamma}R_{0}(\zeta)\mathbb{K}(\xi,a,\mu)R_{0}(\zeta)\zeta \, d\zeta-
\\
-\frac{1}{2\pi i}\oint_{\Gamma}R_{0}(\zeta)\sum_{i=1}^{d}\partial_{i}\A(0,a,\mu)\xi_{i}
R_{0}(\zeta)\Delta_{2}\A(\xi)R_{0}(\zeta)\zeta \, d\zeta-
\\
- \frac{1}{2\pi i}\oint_{\Gamma}R_{0}(\zeta)\Delta_{2}\A(\xi) R_{0}(\zeta)\Delta\A(\xi)R_{0}(\zeta)\zeta\, d\zeta+
\\
+ \frac{1}{2\pi i}\oint_{\Gamma} R(\xi,\zeta)\Delta\A(\xi)R_{0}(\zeta)\Delta\A(\xi)R_{0}(\zeta)\Delta\A(\xi)R_{0}(\zeta)
\zeta \,d\zeta,\ \ |\xi|\le\delta_{0}(a,\mu).
\end{multline}
Since $|\zeta|\le d_{0}/2$ for  $\zeta\in\Gamma$, the relations in \eqref{e2.34+} and estimates \eqref{e2.3}--\eqref{e2.5},  \eqref{e2.32}
imply  \eqref{e2.29} with $C_{2}(a,\mu)$ given by the formula
\begin{multline}
\label{C2}
C_{2}(a,\mu):= \frac{\pi+2}{2\pi} \left(\mu_{+} M_{3}(a)
+  \frac{\mu_{+}^2}{d_0} \left( 3 M_{2}(a)+ M_{3}(a)\Big) \Big( 12 M_1(a) + 3 M_{2}(a)+ M_{3}(a)\right)+\right.
\\
+\left. \frac{\mu_{+}^3}{d_0^2} \left( 6 M_{1}(a) + 3 M_2(a) + M_3(a)\right)^3 \right).
\end{multline}
\end{proof}

\subsection{Approximation of the resolvent of operator $\A(\xi)$}
\label{subsec214}
In this section our goal is to approximate the resolvent of the operator
$\A(\xi,a,\mu)$  in the vicinity of the lower edge of its spectrum.
By \eqref{e1.30} we have
\begin{equation}\label{e2.35}
(\A(\xi,a,\mu)F(\xi)u,u)\ge\mu_{-}C(a)|\xi|^{2}(F(\xi)u,u),\ \
u\in L_{2}(\Omega),\ \ |\xi|\le\delta_{0}(a,\mu).
\end{equation}
Substituting $P+O(|\xi|)$ for $F(\xi)$ in \eqref{e2.35} and considering \eqref{e2.13} we derive the relation
\begin{equation*}
\frac{1}{2}\sum_{i,j=1}^{d}g_{ij}\xi_{i}\xi_{j}\|Pu\|^{2}\ge\mu_{-}C(a)|\xi|^{2}\|Pu\|^{2}+O(|\xi|^{3}),\
\ |\xi|\to0.
\end{equation*}
We then divide this relation by $|\xi|^{2}$, set $u=\mathbf{1}_{\Omega}$ and send $\xi$ to zero in order to obtain
the following inequality:
\begin{equation*}
\frac{1}{2}\sum_{i,j=1}^{d}g_{ij}\theta_{i}\theta_{j}\ge\mu_{-}C(a),\
\ \theta\in\mathbb{S}^{d-1},
\end{equation*}
which finally yields
\begin{equation}\label{e2.36}
\frac{1}{2}\sum_{i,j=1}^{d}g_{ij}\xi_{i}\xi_{j}\ge\mu_{-}C(a)|\xi|^{2},\
\ \xi\in\R^d.
\end{equation}

\begin{theorem}\label{teor2.2}
Assume that conditions \eqref{e1.1}--\eqref{e1.4} hold. Then the following estimate is valid:
\begin{equation}\label{e2.37}
\Big\|(\A(\xi,a,\mu)+\varepsilon^{2}I)^{-1}-\Big(\frac{1}{2}\sum_{i,j=1}^{d}g_{ij}\xi_{i}\xi_{j}+\varepsilon^{2}\Big)^{-1}P\Big\|\le
S(a,\mu)\varepsilon^{-1},\ \ \varepsilon>0,\ \
|\xi|\le\delta_{0}(a,\mu).
\end{equation}
Here the constant $S(a,\mu)$ is given by the formula
\begin{equation*}
S(a,\mu):=\frac{2C_{1}(a,\mu)}{(\mu_{-}C(a))^{1/2}}+\frac{C_{2}(a,\mu)}{(\mu_{-}C(a))^{3/2}}+(2d_{0}/3)^{-1/2}
\end{equation*}
and, in view of \eqref{C(a)}, \eqref{d0}, \eqref{C1} and \eqref{C2}, this constant depends only on
$d$, $\mu_-$, $\mu_+$, $M_1(a)$, $M_2(a)$, $M_3(a)$, ${\mathcal M}(a)$, ${\mathcal A}_\pi(a)$ and
${\mathcal A}_{r(a)}(a)$.
\end{theorem}

\begin{proof}
Combining an evident inequality
\begin{equation}\label{e2.38}
\|(\A(\xi,a,\mu)+\varepsilon^{2}I)^{-1}(I-F(\xi))\|\le
(2d_{0}/3)^{-1},\ \ |\xi|\le\delta_{0}(a,\mu),
\end{equation}
with  \eqref{e1.30} and \eqref{e2.36} results in the following estimates:
\begin{equation}\label{e2.39}
\|(\A(\xi,a,\mu)F(\xi)+\varepsilon^{2}I)^{-1}F(\xi)\|\le(\mu_{-}C(a)|\xi|^{2}+\varepsilon^{2})^{-1},\
\ \varepsilon>0,\ \ |\xi|\le\delta_{0}(a,\mu);
\end{equation}
\begin{equation}\label{e2.40}
\Big\|\Big(\frac{1}{2}\sum_{i,j=1}^{d}g_{ij}\xi_{i}\xi_{j}+\varepsilon^{2}\Big)^{-1}P\Big\|\le(\mu_{-}C(a)|\xi|^{2}+\varepsilon^{2})^{-1},\
\ \varepsilon>0,\ \ |\xi|\le\delta_{0}(a,\mu).
\end{equation}
Since $(\A(\xi,a,\mu)+\varepsilon^{2}I)^{-1}F(\xi)=F(\xi)(\A(\xi,a,\mu)+\varepsilon^{2}I)^{-1}F(\xi)$, after straightforward rearrangements we obtain
\begin{multline}\label{e2.41}
(\A(\xi,a,\mu)+\varepsilon^{2}I)^{-1}F(\xi)-\Big(\frac{1}{2}\sum_{i,j=1}^{d}g_{ij}\xi_{i}\xi_{j}+\varepsilon^{2}\Big)^{-1}P=\\=
F(\xi)(\A(\xi,a,\mu)+\varepsilon^{2}I)^{-1}(F(\xi)-P)+(F(\xi)-P)\Big(\frac{1}{2}\sum_{i,j=1}^{d}g_{ij}\xi_{i}\xi_{j}+\varepsilon^{2}\Big)^{-1}P+\\+
F(\xi)(\A(\xi,a,\mu)+\varepsilon^{2}I)^{-1}\left(\frac{1}{2}\sum_{i,j=1}^{d}g_{ij}\xi_{i}\xi_{j}P-\A(\xi,a,\mu)
F(\xi)\right)\Big(\frac{1}{2}\sum_{i,j=1}^{d}g_{ij}\xi_{i}\xi_{j}+\varepsilon^{2}\Big)^{-1}P,\\
\varepsilon>0,\ \ |\xi|\le\delta_{0}(a,\mu).
\end{multline}
Now \eqref{e2.37} follows from \eqref{e2.28}, \eqref{e2.29} and \eqref{e2.38}--\eqref{e2.41}.
\end{proof}

\section{Homogenization of convolution type operators}
\label{sec3}

We assume that conditions \eqref{e1.1}--\eqref{e1.4} hold and consider a family of nonlocal operators in
$L_{2}(\R^d)$ which are defined as follows:
\begin{equation*}
\A_{\varepsilon}u(x):=\varepsilon^{-d-2}\int_{\R^d}a((x-y)/\varepsilon)\mu(x/\varepsilon,y/\varepsilon)(u(x)-u(y))dy,\
\ x\in\R^d,\ \ u\in L_{2}(\R^d),\ \ \varepsilon>0.
\end{equation*}
Recalling the definition of the coefficients $g_{ij}$ in \eqref{e2.27} we introduce the so called {\it effective operator} which is a symmetric
second order elliptic operator with constant coefficients in $\mathbb R^d$ that reads
\begin{equation*}
\A^{0}:=\frac{1}{2}\sum_{i,j=1}^{d}g_{ij}D_{i}D_{j} = - \operatorname{div} g^0 \nabla.
\end{equation*}
A $(d\times d)$-matrix $g^0$ with the entries $\frac{1}{2}g_{ij}$, $i,j=1,\dots,d$, is called the {\it effective matrix}.
Due to estimate \eqref{e2.36} this matrix is positive definite.
Our main result is given by

\begin{theorem}
\label{teor3.1}
Let conditions \eqref{e1.1}--\eqref{e1.4} be satisfied. Then the following estimate holds:
\begin{equation}\label{e3.1}
\|(\A_{\varepsilon}+I)^{-1}-(\A^{0}+I)^{-1}\|\le
{\mathcal C}(a,\mu)\varepsilon,\ \ \varepsilon>0,
\end{equation}
with the constant ${\mathcal C}(a,\mu)$ defined by the formula
\begin{equation*}
{\mathcal C}(a,\mu):=\frac{2}{\sqrt{\mu_{-}C(a)} \,\delta_{0}(a,\mu)}+S(a,\mu),
\end{equation*}
where $C(a)$ is introduced in \eqref{C(a)}, while $\delta_{0}(a,\mu)$ and $S(a,\mu)$
are defined in Section {\rm\ref{subsec_21}} and in Theorem {\rm\ref{teor2.2}}, respectively.
The constant ${\mathcal C}(a,\mu)$  depends only on $\mu_-$, $\mu_+$,
$d$, $M_1(a)$, $M_2(a)$, $M_3(a)$, ${\mathcal M}(a)$, ${\mathcal A}_\pi(a)$ and ${\mathcal A}_{r(a)}(a)$.
\end{theorem}

\begin{proof}
Consider a family of scaling transformations of the form
\begin{equation*}
T_{\varepsilon}u(x):=\varepsilon^{d/2}u(\varepsilon x),\ \
x\in\R^d,\ \ u\in L_{2}(\R^d), \ \ \eps >0.
\end{equation*}
Notice that $T_\eps$ are unitary operators and that
$T_{\varepsilon}\A_{\varepsilon}T_{\varepsilon}^{*}=\varepsilon^{-2}\A(a,\mu)$.
Consequently, by \eqref{e1.8} we have
\begin{multline}\label{e3.2}
(\A_{\varepsilon}+I)^{-1}=T_{\varepsilon}^{*}\varepsilon^{2}(\A+\varepsilon^{2}I)^{-1}T_{\varepsilon}=
T_{\varepsilon}^{*}\mathcal{G}^{*}[\varepsilon^{2}(\A(\xi,a,\mu)+\varepsilon^{2}I)^{-1}]\mathcal{G}T_{\varepsilon}=\\=
T_{\varepsilon}^{*}\mathcal{G}^{*}[\varepsilon^{2}\mathbf{1}_{|\xi|\le\delta_{0}(a,\mu)}(\A(\xi,a,\mu)+\varepsilon^{2}I)^{-1}]\mathcal{G}T_{\varepsilon}+\\+
T_{\varepsilon}^{*}\mathcal{G}^{*}[\varepsilon^{2}\mathbf{1}_{|\xi|>\delta_{0}(a,\mu)}(\A(\xi,a,\mu)+\varepsilon^{2}I)^{-1}]\mathcal{G}T_{\varepsilon},\ \ \varepsilon>0.
\end{multline}
Estimate \eqref{e1.30} and the relations in \eqref{e3.2}  imply the inequality
\begin{equation}\label{e3.3}
\Big\|(\A_{\varepsilon}+I)^{-1}-
T_{\varepsilon}^{*}\mathcal{G}^{*}[\varepsilon^{2}\mathbf{1}_{|\xi|\le\delta_{0}(a,\mu)}(\A(\xi,a,\mu)+
\varepsilon^{2}I)^{-1}]\mathcal{G}T_{\varepsilon}\Big\|\le\frac{\varepsilon}{\sqrt{\mu_{-}C(a)}\,\delta_{0}(a,\mu)},\
\ \varepsilon>0.
\end{equation}
It also follows from estimates \eqref{e2.37} and \eqref{e3.3} that
\begin{multline}\label{e3.4}
\Big\|(\A_{\varepsilon}+I)^{-1}-T_{\varepsilon}^{*}\mathcal{G}^{*}[\varepsilon^{2}\mathbf{1}_{|\xi|\le\delta_{0}(a,\mu)}
\Big(\frac{1}{2}\sum_{i,j=1}^{d}g_{ij}\xi_{i}\xi_{j}+\varepsilon^{2}\Big)^{-1}P]\mathcal{G}T_{\varepsilon}\Big\|\le\\\le
\frac{\varepsilon}{\sqrt{\mu_{-}C(a)}\,\delta_{0}(a,\mu)} + S(a,\mu) \varepsilon,\ \ \varepsilon>0.
\end{multline}
As was shown for example in \cite{Sk} or \cite{BSu1}, for any $H\in L_{\infty}(\widetilde\Omega)$
the relation $\mathcal{G}^{*}[H(\xi)P]\mathcal{G}=\Phi^{*}[\mathbf{1}_{\widetilde\Omega}(\xi)H(\xi)]\Phi$
holds, where the symbol $\Phi$ stands for the Fourier transform.
It remains to note that estimate \eqref{e3.1} is a direct  consequence of \eqref{e2.36} and \eqref{e3.4}.
\end{proof}

We conclude this section with  two almost obvious remarks.

\begin{remark}
Theorem {\rm \ref{teor3.1}} remains valid if, in condition  \eqref{e1.3}, the cubic lattice ${\mathbb Z}^d$
is replaced with an arbitrary periodic lattice $\mathbf{\Gamma}$ in $\mathbb R^d$. Under such replacement,
the constant $\mathcal C$ in estimate \eqref{e3.1}  depends not only on the functions $a$ and $\mu$, but also on
the lattice parameters. The coefficients of the homogenized operator $\mathbb{A}^{0}$ are still determined by
relations \eqref{e2.23} and \eqref{e2.27}, where $\Omega$ is the cell of the lattice $\mathbf{\Gamma}$.
\end{remark}

\begin{remark}
Suppose that, in Theorem   {\rm \ref{teor3.1}}, condition \eqref{e1.4} 
is modified as follows:
\begin{equation*}
M_{k}(a):=\int_{\mathbb{R}^d}|x|^{k}a(x)\,dx< \infty,\ \ \hbox{for } k=1,2, k_0
\end{equation*}
with $k_0$ that satisfies the inequalities $2<k_0<3$.  Then instead of \eqref{e3.1} we have
\begin{equation*}
\|(\mathbb{A}_{\eps}+I)^{-1}-(\mathbb{A}^{0}+I)^{-1}\|\le C\eps^{k_0-2},\ \ \eps>0.
\end{equation*}
Here the constant $C$ depends on $a$, $\mu$ and $k_0$.
\end{remark}


\section{Appendix}
\label{sec_app}

\subsection{Schur test and Hadamard lemma}
\label{subsec_schur}

In this section for the reader convenience we recall the statements of
two classical results.
We begin with a simple version of the Schur test.

\begin{lemma}[\bf Schur test]
\label{lemma4.1}
Let $({\mathcal X},d\rho)$ and  $({\mathcal Y},d\tau)$ be separable measurable spaces with $\sigma$-finite measures, and assume that
 $\B:L_{2}({\mathcal Y},d\tau)\to L_{2}({\mathcal X},d\rho)$ is a linear integral operator with a kernel
$b(x,y)$, $x\in {\mathcal X}$, $y\in \mathcal{Y}$. Assume furthermore that the following conditions hold:
\begin{align*}
\alpha:=\rho{\text -}\!\esssup_{x\in\mathcal{X}}\intop_{\mathcal{Y}}|b(x,y)| \, d\tau(y)<+\infty,
\\
\beta:=\tau{\text -}\!\esssup_{y\in \mathcal{Y}}\intop_{\mathcal{X}}| b(x,y)| \, d\rho(x)<+\infty.
\end{align*}
Then the operator $\B$ is bounded and satisfies the estimate $\|\B\| \le(\alpha\beta)^{1/2}$.
\end{lemma}

Recalling the notation $\Omega:=[0,1)^{d}$ and letting $\Delta:=[-1,1]^{d}$
we derive from the latter theorem  the following result:

\begin{corollary}
\label{corollary4.2}
For any functions
$\varphi\in L_{1}(\Delta)$ and  $\psi\in L_{\infty}(\Omega\times\Omega)$ the operator
$\B\,:\,L_{2}(\Omega)\,\to\,L_{2}(\Omega)$ given by
\begin{equation*}
\B u(x):=\int_{\Omega}\varphi(x-y)\psi(x,y)u(y)dy,\ \ x\in\Omega,\
\ u\in L_{2}(\Omega),
\end{equation*}
is compact and satisfies the estimate
\begin{equation}\label{e4.1}
\|\B\|\le\|\varphi\|_{L_{1}(\Delta)}\|\psi\|_{L_{\infty}(\Omega\times\Omega)}.
\end{equation}
\end{corollary}

\begin{proof}
By Lemma \ref{lemma4.1} the operator $\B$ is bounded, and estimate
 \eqref{e4.1} holds.  To show that $\B$ is compact we choose a sequence
 of functions $\{\varphi_{n}\}_{n=1}^{\infty}\subset C_{0}^{\infty}(\Delta)$
 such that $\|\varphi_{n}-\varphi\|_{L_1(\Delta)}\to 0$ as $n\to\infty$.
 Denote by $\B_{n}$ the operators with the kernels
$\varphi_{n}(x-y)\psi(x,y)$, $x,y\in\Omega$. By  \eqref{e4.1}
we have
$$
\|\B_{n}-\B\|\le\|\varphi_{n}-\varphi\|_{L_1(\Delta)}\|\psi\|_{L_\infty(\Omega \times \Omega)}\to 0, \quad n\to\infty.
$$
Since for each $n\in\mathbb N$ the kernel
$\varphi_{n}(x-y)\psi(x,y)$, $x,y\in\Omega$, is bounded and  $\Omega$ is a compact set,
then $\B_n$ are Hilbert-Schmidt operators. Therefore, $\B$ is compact.
\end{proof}
Next, we formulate a version of the Hadamard lemma which is used
in this work.
\begin{lemma}[Hadamard lemma]
\label{Adamar}
Let a function $f(x)$, $x\in\R^d$, be three times continuously differentiable.
Then it admits the following representation:
\begin{multline*}
f(y)=f(0)+\sum_{i=1}^{d}y_{i}\partial_{i}f(0)+\frac{1}{2}\sum_{i=1}^{d}\sum_{j=1}^{d}y_{i}y_{j}\partial_{i}\partial_{j}f(0)+
\\
+ \sum_{i=1}^{d}\sum_{j=1}^{d}\sum_{k=1}^{d}y_{i}y_{j}y_{k}\intop_{0}^{1}ds_{1}s_{1}^{2}\intop_{0}^{1}ds_{2}\,s_{2}\intop_{0}^{1}ds_{3} \,\partial_{i}\partial_{j}\partial_{k}f(s_{1}s_{2}s_{3}y),\ \
y\in\R^d.
\end{multline*}
\end{lemma}

\subsection{Other methods of determining the coefficients in \eqref{e2.8}.}
\label{subsec4_2_other}
We call the method of computing the coefficients in \eqref{e2.8} that was used in Section \ref{subsec_22} {\sl the first method}.
In this section we describe two other methods.

\noindent
{\bf Second methos}.  For sufficiently small $\xi\in\mathbb{R}^{d}$ we define
\begin{equation*}
\varphi(\xi):=F(\xi)\mathbf{1}_{\Omega},\  \ e(\xi):=\varphi(\xi)/\|\varphi(\xi)\|,\ \ \lambda(\xi):=(\mathbb{A}(\xi,a,\mu)e(\xi),e(\xi)).
\end{equation*}
Then $\lambda(\xi)$ is the lowest eigenvalue of the operator $\mathbb{A}(\xi,a,\mu)$,  and $e(\xi)$ is the corresponding normalized
eigenfunction. The eigenvalue $\lambda(\xi)$ is simple. Therefore, by the Riesz formula, see  \eqref{e2.6}, the functions
$\varphi(\xi)$, $e(\xi)$ and $\lambda(\xi)$ are  $C^{3}$ smooth functions of $\xi$ in some neighbourhood of zero,
and the Taylor formula yields
\begin{equation}\label{dop4.1}
\lambda(\xi)=\lambda_{0}+\sum_{i=1}^{d}\lambda_{i}\xi_{i}+\frac{1}{2}\sum_{i,j=1}^{d}\lambda_{ij}\xi_{i}\xi_{j}+O(|\xi|^{3}),\ \ \xi\to 0;
\end{equation}
\begin{equation}\label{dop4.2}
e(\xi)=e_{0}+\sum_{i=1}^{d}e_{i}\xi_{i}+\frac{1}{2}\sum_{i,j=1}^{d}e_{ij}\xi_{i}\xi_{j}+O(|\xi|^{3}),\ \ \xi\to 0.
\end{equation}
Since $\lambda(0)=0$ and $\lambda(\xi)\ge 0$ in the vicinity of zero, then $\lambda_{i}=0$, $i=1,\dots,d$.
We also have $e_{0}=e(0)=\mathbf{1}_{\Omega}$. Combining \eqref{dop4.1}, \eqref{dop4.2} and the equality
\begin{equation*}
\mathbb{A}(\xi,a,\mu)F(\xi)=\lambda(\xi)(\cdot,e(\xi))e(\xi),
\end{equation*}
we obtain the following formulas for the coefficients in \eqref{e2.8}:
\begin{equation}\label{dop4.3}
G_{0}=0,\ \ G_{i}=0,\ \ i=1,\dots,d,\ \ G_{ij}=\lambda_{ij}(\cdot,\mathbf{1}_{\Omega})\mathbf{1}_{\Omega}=\lambda_{ij}P,\ \ i,j=1\dots,d.
\end{equation}
It remains to compute $\lambda_{ij}$.
Inserting the expansions \eqref{e2.2}, \eqref{dop4.1} and \eqref{dop4.2} in the relation
\begin{equation*}
\mathbb{A}(\xi,a,\mu)e(\xi)=\lambda(\xi)e(\xi),
\end{equation*}
and comparing the coefficients of the first and the second order terms in the resulting relation we conclude that
\begin{equation}\label{dop4.4}
\mathbb{A}(0,a,\mu)e_{i}+\partial_{i}\mathbb{A}(0,a,\mu)\mathbf{1}_{\Omega}=0,\ \ i=1,\dots,d;
\end{equation}
\begin{equation}\label{dop4.5}
\mathbb{A}(0,a,\mu)e_{ij}+\partial_{i}\mathbb{A}(0,a,\mu)e_{j}+\partial_{j}\mathbb{A}(0,a,\mu)e_{i}+\partial_{i}\partial_{j}\mathbb{A}(0,a,\mu)\mathbf{1}_{\Omega}=\lambda_{ij}\mathbf{1}_{\Omega},\ \ i,j=1,\dots,d.
\end{equation}
From \eqref{dop4.4} we deduce that
\begin{equation}\label{dop4.6}
\partial_{i}\mathbb{A}(0,a,\mu)\mathbf{1}_{\Omega}=P^{\bot}\partial_{i}\mathbb{A}(0,a,\mu)\mathbf{1}_{\Omega},\ \ i=1,\dots,d;
\end{equation}
\begin{equation}\label{dop4.7}
P^{\bot}e_{i}=-P^{\bot}\mathbb{A}(0,a,\mu)^{-1}P^{\bot}\partial_{i}\mathbb{A}(0,a,\mu)\mathbf{1}_{\Omega},
\end{equation}
and \eqref{dop4.5} yields the relation
\begin{multline}\label{dop4.8}
\lambda_{ij}=(\lambda_{ij}\mathbf{1}_{\Omega},\mathbf{1}_{\Omega})=(\mathbb{A}(0,a,\mu)e_{ij},\mathbf{1}_{\Omega})+(\partial_{i}\mathbb{A}(0,a,\mu)e_{j},\mathbf{1}_{\Omega})+(\partial_{j}\mathbb{A}(0,a,\mu)e_{i},\mathbf{1}_{\Omega})+\\+(\partial_{i}\partial_{j}\mathbb{A}(0,a,\mu)\mathbf{1}_{\Omega},\mathbf{1}_{\Omega}),\ \ i,j=1,\dots,d.
\end{multline}
Considering \eqref{dop4.6} and the equality
$(\mathbb{A}(0,a,\mu)e_{ij},\mathbf{1}_{\Omega})=(e_{ij},\mathbb{A}(0,a,\mu)\mathbf{1}_{\Omega})=0$
we obtain
\begin{multline*}
(\partial_{i}\mathbb{A}(0,a,\mu)e_{j},\mathbf{1}_{\Omega})=(e_{j},\partial_{i}\mathbb{A}(0,a,\mu)\mathbf{1}_{\Omega})=(P^{\bot}e_{j},\partial_{i}\mathbb{A}(0,a,\mu)\mathbf{1}_{\Omega})=\\=(\partial_{i}\mathbb{A}(0,a,\mu)P^{\bot}e_{j},\mathbf{1}_{\Omega}),\ \ i,j=1,\dots,d.
\end{multline*}
Finally, this relation combined with \eqref{dop4.7} and \eqref{dop4.8} implies that
\begin{multline}\label{dop4.9}
\lambda_{ij}=((\partial_{i}\partial_{j}\A(0,a,\mu)-\partial_{i}\A(0,a,\mu)P^{\perp}\A(0,a,\mu)^{-1}P^{\perp}\partial_{j}\A(0,a,\mu))\mathbf{1}_{\Omega},\mathbf{1}_{\Omega})-\\-(\partial_{j}\A(0,a,\mu)P^{\perp}\A(0,a,\mu)^{-1}P^{\perp}\partial_{i}\A(0,a,\mu))\mathbf{1}_{\Omega},\mathbf{1}_{\Omega}),\ \ i,j=1,\dots,d,
\end{multline}
and \eqref{e2.19} follows from \eqref{dop4.9} and \eqref{e2.19}.

\noindent
{\bf Third method}.
The first method described above is not efficient for obtaining the higher order terms of the asymptotics of the operator
  $\mathbb{A}(\xi,a,\mu)F(\xi)$ and thus for obtaining a more precise approximation of the resolvent $({\mathbb A}_\eps +I)^{-1}$.
  The second method fails to work in the case of a system of equations if the lowest eigenvalue of the operator $\mathbb{A}(0,a,\mu)$
  is not simple.

  Here we consider the third method of computing the coefficients in \eqref{e2.8} that does not have the mentioned demerits.
The coefficients of expansion \eqref{e2.8} are given by the contour integrals \eqref{dop2.1}--\eqref{dop2.3}.
The computation is based on the following representation of the resolvent of operator $\mathbb{A}(0,a,\mu)$:
\begin{equation}\label{dop4.10}
R_{0}(\zeta)=R_{0}(\zeta)P+R_{0}(\zeta)P^{\bot}=-\frac{1}{\zeta}P+R_{0}(\zeta)P^{\bot},\ \ \zeta\in\Gamma.
\end{equation}
We substitute the right-hand side of \eqref{dop4.10} for $R_{0}(\zeta)$ in the contour integrals in \eqref{dop2.1} and \eqref{dop2.2}
and use the fact that the operator-function $R_{0}^{\bot}(\zeta):=R_{0}(\zeta)P^{\bot}$ is holomorphic inside the contour $\Gamma$.
This yields the relations $G_{0}=0$ and
\begin{multline*}
G_{i}=\frac{-1}{2\pi i}\oint_{\Gamma} \Bigl(-\frac{1}{\zeta}P+R_{0}^{\bot}(\zeta)\Bigr)\partial_{i}\mathbb{A}(0,a,\mu)
\Bigl(-\frac{1}{\zeta}P+R_{0}^{\bot}(\zeta)\Bigr)\zeta \,d\zeta=\\=
\frac{-1}{2\pi i}\oint_{\Gamma} \Bigl(-\frac{1}{\zeta}P\Bigr)\partial_{i}\mathbb{A}(0,a,\mu)\Bigl(-\frac{1}{\zeta}P\Bigr)\zeta \, d\zeta=-P\partial_{i}\mathbb{A}(0,a,\mu) P,\ \ i=1,\dots,d.
\end{multline*}
As an immediate  consequence of \eqref{e2.12} we have
\begin{equation}\label{dop4.11}
P\partial_{i}\mathbb{A}(0,a,\mu) P=0,\ \ i=1,\dots,d.
\end{equation}
We insert expansion \eqref{dop4.10} in \eqref{dop2.3}, take into account the fact that the operator-function $R_{0}^{\bot}(\zeta)$
is holomorphic inside the contour $\Gamma$, and use the equality
$$
\frac{1}{2\pi i}\oint_{\Gamma}  \frac{1}{\zeta}R_{0}^{\bot}(\zeta) \,d\zeta = P^\perp
{\mathbb A}(0,a,\mu)^{-1}P^{\perp}
$$
and \eqref{dop4.11}. This leads to the desired representation \eqref{e2.19}.

\subsection{Estimates of the constants $\mathcal{A}_{r}(a)$ and $\mathcal{M}(a)$.}
\label{Append2}
This section focuses on estimates of the constants $\mathcal{A}_{r}(a)$, \hbox{$r>0$}, and  $\mathcal{M}(a)$ defined in
 \eqref{e1.17a}--\eqref{e1.20} in terms of $\|a\|_{L_1}$, $\|a\|_{L_2}$ and $M_{k}(a)$, $k=1,2,3$.
Here we assume that conditions \eqref{e1.1}--\eqref{e1.4} are fulfilled and that in addition $a\in L_{2}(\R^d)$.
The following upper bounds are trivial:
\begin{align}
\label{e1.33}
\mathcal{A}_{r}(a) &\le 2\|a\|_{L_1},\ \ r>0;
\\
\label{e1.34}
\mathcal{M}(a) &\le M_{2}(a).
\end{align}
We turn to the lower bounds.

\begin{lemma}\label{lemma1.6}
Let conditions \eqref{e1.1} and \eqref{e1.4} hold. Then the following lower bound is valid:
\begin{equation}\label{e1.35}
\int_{|z|\le \rho}a(z)\, dz\ge\frac{7}{8}\|a\|_{L_1},\ \ \rho \ge \rho_{0}(a):=2M_{3}^{1/3}(a)\|a\|_{L_1}^{-1/3}.
\end{equation}
\end{lemma}

\begin{proof}
The desired estimate \eqref{e1.35}  follows from the inequalities
\begin{equation*}
\int_{|z|>\rho}a(z)\,dz\le \rho^{-3}\int_{|z|>\rho}|z|^{3}a(z) \,dz\le
\rho^{-3}M_{3}(a),\ \ \rho >0.
\end{equation*}
\end{proof}

Define a set
$\Pi_{r}^{\theta}(a):=\{z\in\R^d\,:\,|z|\le \rho_{0}(a),\ \ |\langle
z,\theta\rangle|\le r\},\ \ r>0,\ \ \theta\in\mathbb{S}^{d-1}.$
It is straightforward to check that
\begin{equation}\label{e1.36}
\operatorname{mes}\Pi_{r}^{\theta}(a)\le
2r\varkappa_{d-1} \rho_{0}^{d-1}(a),
\end{equation}
where $\varkappa_{d-1}$ is the volume of a unit ball in $\R^{d-1}$.

\begin{lemma}\label{lemma1.7}
Let conditions \eqref{e1.1} and \eqref{e1.4} be satisfied, and assume that $a\in L_{2}(\R^d)$.
Then the following lower bound holds:
\begin{multline}\label{e1.37}
\int_{B_{\rho_{0}(a)}(0)\setminus\Pi_{r}^{\theta}(a)}a(z)dz\ge\frac{1}{2}\|a\|_{L_1},\
\ |\theta|=1,\ \ r\in(0,r_{0}(a)],\\
r_{0}(a):=\Big(\frac{3}{8}\|a\|_{L_1}\Big)^{2}\Big(2\varkappa_{d-1}\rho_{0}^{d-1}(a)\|a\|_{L_2}^{2}\Big)^{-1}.
\end{multline}
Here the constant $\rho_{0}(a)$ has been defined in \eqref{e1.35}.
\end{lemma}

\begin{proof}
By \eqref{e1.36} and the Cauchy-Schwartz inequality, we obtain the estimates
\begin{equation}\label{e1.38}
\int_{\Pi_{r}^{\theta}(a)}a(z)dz\le\|a\|_{L_2}\Big(\operatorname{meas}\Pi_{r}^{\theta}(a)\Big)^{1/2}\le\Big(2r\varkappa_{d-1} \rho_{0}^{d-1}(a)\Big)^{1/2}\|a\|_{L_2},\
\ r>0,\ |\theta|=1.
\end{equation}
Now \eqref{e1.37} is a consequence of \eqref{e1.35} and \eqref{e1.38}.
\end{proof}

The next statement gives a lower bound for the constant $\mathcal{M}(a)$.

\begin{lemma}\label{lemma1.8}
Under the conditions of Lemma {\rm\ref{lemma1.7}} we have
\begin{equation}\label{e1.39}
\mathcal{M}(a)\ge\frac{1}{2}\|a\|_{L_1}r_{0}(a)^{2},
\end{equation}
where $r_{0}(a)$ has been defined in  {\rm(\ref{e1.37})}.
\end{lemma}

\begin{proof}
By \eqref{e1.37} for any $\theta\in\mathbb{S}^{d-1}$ the following estimates hold:
\begin{equation*}
M_{a}(\theta)=\int_{\R^d}a(z)\langle z,\theta\rangle^{2}dz\ge
\int_{B_{\rho_{0}(a)}(0)\setminus\Pi_{r_{0}(a)}^{\theta}(a)}a(z)\langle
z,\theta\rangle^{2}dz\ge \frac{1}{2}\|a\|_{L_1} r_{0}^{2}(a).
\end{equation*}
Combining these estimates with \eqref{e1.20} yields \eqref{e1.39}.
\end{proof}

\begin{lemma}\label{lemma1.9}
Let the conditions of Lemma {\rm\ref{lemma1.7}} hold. Then the function  $\hat A(y)$ defined in \eqref{e1.17a} admits
the estimate
\begin{equation*}
\hat A(y)\ge\frac{1}{8}\|a\|_{L_1}r_{0}^{2}(a)|y|^{2},\ \
|y|\le(2\rho_{0}(a))^{-1};
\end{equation*}
here $\rho_{0}(a)$ and  $r_{0}(a)$ have been introduced in \eqref{e1.35} and
\eqref{e1.37}, respectively.
\end{lemma}

\begin{proof}
Consider the set of $y\in\R^d$ such that $\displaystyle 0<|y|\le(2\rho_{0}(a))^{-1}$, and denote
by  $\theta$ the vector $y/|y|$.
Taking into account an elementary estimate $1-\cos t\ge\frac{1}{4}t^{2}$,
$|t|\le 1/2$, by Lemma  \ref{lemma1.7}, we obtain
\begin{multline*}
\hat A(y)=\int_{\R^d}(1-\cos\langle z,y\rangle)a(z)\,dz\ge
\int_{B_{\rho_{0}(a)}}(1-\cos\langle z,y\rangle)a(z)\,dz\ge
\frac{1}{4}\int_{B_{\rho_{0}(a)}}\langle z,y\rangle^{2}a(z)\,dz\ge
\\
\ge\frac{1}{4}|y|^{2}\int_{B_{\rho_{0}(a)}\setminus\Pi_{r_{0}(a)}^{\theta}(a)}\langle
z,\theta\rangle^{2}a(z)\, dz\ge\frac{1}{8}|y|^{2}\|a\|_{L_1}r_{0}^{2}(a).
\end{multline*}
\end{proof}

Introducing the notation
\begin{equation*}
Q_{\tau}^{y}(a):=\{z\in B_{\rho_{0}(a)}(0)\,:\,1-\cos\langle
z,y\rangle\le\tau\},\ \ |y|\ge(2\rho_{0}(a))^{-1},\ \ \tau\in(0,1/2],
\end{equation*}
one can easily check that the following relation is valid:
\begin{multline}\label{e1.41}
Q_{\tau}^{y}(a)=\bigcup_{n\in\mathbb{Z}}Q_{\tau,n}^{y}(a),\\
Q_{\tau,n}^{y}(a):=\{z\in B_{\rho_{0}(a)}(0)\,:\,2\pi
n-\operatorname{arccos}(1-\tau)\le\langle z,y\rangle\le 2\pi
n+\operatorname{arccos}(1-\tau)\},\\ |y|\ge(2 \rho_{0}(a))^{-1},\ \
\tau\in(0,1/2].
\end{multline}

\begin{lemma}\label{lemma1.10-}
Assume that conditions \eqref{e1.1} and \eqref{e1.4} are fulfilled.
Then the following estimates take place:
\begin{equation}\label{e1.42}
\#\{n\in\mathbb{Z}\,:\,Q_{\tau,n}^{y}(a)\not=\varnothing\}\le\pi^{-1}|y| \rho_{0}(a)+2,\
\ |y|\ge(2 \rho_{0}(a))^{-1},\ \ \tau\in(0,1/2];
\end{equation}
\begin{equation}\label{e1.43}
\operatorname{mes}Q_{\tau,n}^{y}(a)\le2|y|^{-1}\operatorname{arccos}(1-\tau)\varkappa_{d-1} \rho_{0}^{d-1}(a),\
\ |y|\ge(2 \rho_{0}(a))^{-1},\ \ \tau\in(0,1/2].
\end{equation}
\end{lemma}

\begin{proof}
By the definition of $Q_{\tau,n}^{y}(a)$, for any vector $z\in Q_{\tau,n}^{y}(a)$ we have
\begin{equation*}
\left\{
\begin{matrix}
(2n-1)\pi\le\langle z,y\rangle\le(2n+1)\pi,\\[2mm]
-|y| \rho_{0}(a)\le\langle z,y\rangle\le |y| \rho_{0}(a).
\end{matrix}
\right.
\end{equation*}
Consequently, if the set $Q_{\tau,n}^{y}(a)$ is not empty, then the following inequalities hold:
\begin{equation*}
\left\{
\begin{matrix}
-|y| \rho_{0}(a)\le (2n+1)\pi,\\[2mm]
(2n-1)\pi\le |y| \rho_{0}(a).
\end{matrix}
\right.
\end{equation*}
Thus for such sets $Q_{\tau,n}^{y}(a)$ we obtain the estimate $|n|\le(2\pi)^{-1}|y|\rho_{0}(a)+1/2$,
 which in turn implies \eqref{e1.42}.

In order to estimate the measure of $Q_{\tau,n}^{y}(a)$, one can choose an orthonormal basis
in $\mathbb R^d$  in such a way that the first coordinate vector is directed along the vector $y$.
Under this choice, for any $z\in Q_{\tau,n}^{y}(a)$, we obtain the inequalities
\begin{equation*}
2\pi n-\operatorname{arccos}\,(1-\tau)\le z_{1}|y|\le 2\pi
n+\operatorname{arccos}\,(1-\tau),\ \ z_{2}^{2}+\dots+z_{d}^{2}\le \rho^{2}_{0}(a),
\end{equation*}
This yields \eqref{e1.43}.
\end{proof}

From \eqref{e1.41}--\eqref{e1.43} we deduce
\begin{lemma}\label{lemma1.10}
Under conditions \eqref{e1.1} and \eqref{e1.4} the following estimate is satisfied:
\begin{multline}\label{e1.44}
\operatorname{mes}Q_{\tau}^{y}(a)\le\operatorname{arccos}(1-\tau)N(a),\\
\hbox{\rm with}\ \ N(a):=\Big(\frac{2}{\pi}+8\Big) \varkappa_{d-1} \rho_{0}^{d}(a),\ \
|y|\ge(2\rho_{0}(a))^{-1},\ \ \tau\in(0,1/2];
\end{multline}
here the quantity $\rho_{0}(a)$ is defined in  \eqref{e1.35}.
\end{lemma}
Using the Cauchy-Schwartz inequality we derive from \eqref{e1.44} that
\begin{equation}\label{e1.45}
\int_{Q_{\tau}^{y}(a)}a(z)\, dz\le\|a\|_{L_2}\Big(\operatorname{arccos}(1-\tau)N(a)\Big)^{1/2},\
\ |y|\ge(2 \rho_{0}(a))^{-1},\ \ \tau\in(0,1/2].
\end{equation}
Now our aim is to construct a semi-open interval $(0,\tau_0(a)]\subset(0,1/2]$ such that for all  $\tau\in(0,\tau_0(a)]$
the right-hand side in \eqref{e1.45} does not exceed $\frac{3}{8}\|a\|_{L_{1}}$.
To this end we can choose
\begin{equation}\label{e4.14+}
\tau_0(a):=\left\{
\begin{array}{ll}
1/2,\ \ \  \ \ \text{ if } \Big(\frac{3}{8}\|a\|_{L_1}\|a\|_{L_2}^{-1}\Big)^{2} >\pi N(a);\\[3mm]
\min \Big\{1/2,
1-\cos\Big(N(a)^{-1}\Big(\frac{3}{8}\|a\|_{L_1}\|a\|_{L_2}^{-1}\Big)^{2}\Big) \Big\},
\hbox{ otherwise.}
\end{array}
\right.
\end{equation}
Considering \eqref{e1.35}, \eqref{e1.45} and \eqref{e4.14+} we obtain
\begin{lemma}\label{lemma1.11}
Under the assumptions of Lemma {\rm\ref{lemma1.7}} the following estimate holds:
\begin{equation}\label{e1.46}
\int_{B_{\rho_{0}(a)}(0)\setminus
Q_{\tau}^{y}(a)}a(z)dz\ge\frac{1}{2}\|a\|_{L_1},\ \ |y|\ge(2\rho_{0}(a))^{-1},\ \
0<\tau\le\tau_{0}(a);
\end{equation}
here $\rho_{0}(a)$ and $N(a)$ are defined in  {\rm(\ref{e1.35})} and
{\rm(\ref{e1.44})}, respectively.
\end{lemma}
As a consequence of the latter statement we have
\begin{lemma}\label{lemma1.12}
Under the conditions of Lemma {\rm\ref{lemma1.7}} the following estimate
is satisfied:
\begin{equation*}
\hat A(y)\ge\frac{1}{2}\|a\|_{L_1}\tau_{0}(a),\ \
|y|\ge(2\rho_{0}(a))^{-1};
\end{equation*}
here the constants $\rho_{0}(a)$ and $\tau_{0}(a)$ are introduced in \eqref{e1.35} and \eqref{e4.14+}, respectively.
\end{lemma}
Finally, recalling the definition of  $r_{0}(a)$ in \eqref{e1.37} and that of
 $\tau_{0}(a)$ in \eqref{e4.14+}, and taking into account   \eqref{e1.33} and \eqref{e1.34},
 by Lemmas  \ref{lemma1.8}, \ref{lemma1.9} and \ref{lemma1.12}, we obtain
\begin{proposition}\label{prop1.13}
Let conditions \eqref{e1.1} and \eqref{e1.4} be fulfilled, and assume that
$a\in L_{2}(\R^d)$.  Then the following estimates hold:
\begin{equation*}
\frac{1}{2}\|a\|_{L_1}r_{0}(a)^{2}\le\mathcal{M}(a)\le M_{2}(a);
\end{equation*}
\begin{equation*}
\min\{\frac{1}{8}\|a\|_{L_1}r_{0}(a)^{2}r^{2},\frac{1}{2}\|a\|_{L_1}\tau_{0}(a)\}\le\mathcal{A}_{r}(a)\le 2\|a\|_{L_1},\
\ r>0.
\end{equation*}
\end{proposition}

\begin{remark}
If the assumptions of Proposition {\rm \ref{prop1.13}} hold, then $\mathcal{M}(a)$ can be estimated from below
by a positive constant that depends only on $d$, $\|a\|_{L_1}$, $\|a\|_{L_2}$
 and $M_3(a)$,  while $\mathcal{A}_r(a)$ admits a lower bound by a positive constant that depends on the same
 parameters and $r$. Furthermore, under the same assumptions,  the constant  ${\mathcal C}(a,\mu)$
 introduced in Theorem {\rm \ref{teor3.1}}  satisfies the estimate ${\mathcal C}(a,\mu)\geq \tilde C>0$, where
  $\tilde C$ depends only on
  $ \mu_-, \; \mu_+, \; d,\;  \|a\|_{L_1},\; \|a\|_{L_2}$  and $M_j(a), \,j=1,\,2,\,3$.
  \end{remark}

\subsection*{Acknowledgements}
The research of A. Piatnitski and E. Zhizhina was supported by the Ministry of Science and Higher Education of the Russian Federation,   agreements 075-15-2022-287 and  075-15-2022-289, respectively. 

The research of V. Sloushch  and T. Suslina was supported by Russian Science Foundation (project 22-11-00092,
 https://rscf.ru/project/22-11-00092/).

\end{document}